\documentclass[12pt,a4paper]{article}

\usepackage{amsmath,amssymb,amsthm,mathtools}
\usepackage[margin=2.5cm]{geometry}

\usepackage{parskip}
\usepackage{enumerate}
\usepackage{enumitem}
\setlist[enumerate, 1]{label=(\roman*)}
\setlist[itemize]{leftmargin=1.5em}
\setlist[description]{leftmargin=1em}

\usepackage{overpic}
\usepackage{contour}
\contourlength{1pt}
\usepackage{graphicx}
\graphicspath{ {figures/} }
\DeclareGraphicsExtensions{.pdf, .jpg, .jpeg, .png}

\usepackage[hidelinks, backref=page]{hyperref}
\usepackage[format=plain,labelsep=period,justification=justified,font=small,labelfont=bf, margin=1em]{caption}

\makeatletter 
\newcommand\nobreakpar{\par\nobreak\@afterheading} 
\makeatother

\numberwithin{equation}{section}

\usepackage{thmtools}
\declaretheorem[style=plain,name=Theorem,qed={\tiny$\blacksquare$},numberwithin=section]{theorem}

\declaretheorem[style=definition,name=Definition,sibling=theorem,qed={\tiny$\blacksquare$}]{definition}
\declaretheorem[style=plain,name=Lemma,sibling=theorem,qed={\tiny$\blacksquare$}]{lemma}

\declaretheorem[style=definition,name=Remark,sibling=theorem,qed={\tiny$\blacksquare$}]{remark}
\declaretheorem[style=definition,name=Question,sibling=theorem,qed={\tiny$\blacksquare$}]{question}

\declaretheorem[style=plain,name=Proposition,sibling=theorem,qed={\tiny$\blacksquare$}]{proposition}

\usepackage{manfnt} 
\newcommand*{\cube}{\mbox{\mancube}}

\newcommand{\mobq}{{\mathcal{M}}}

\newcommand{\lieq}{{\mathcal{L}}}

\newcommand{\Z}{\mathbb{Z}}
\newcommand{\R}{\mathbb{R}}

\newcommand{\RP}{\mathbb{R}\mathrm{P}}

\newcommand{\set}[2]{\left\{#1\ \middle| \ #2 \right\}}
\newcommand{\pol}{\perp}
\newcommand{\sca}[1]{\left<#1\right>}

\title{Multi-dimensional consistency of principal binets}

\author{
  Niklas C. Affolter\thanks{Institut für Diskrete Mathematik und Geometrie, TU Wien, Austria; and
	Institut für Mathematik, TU Berlin, Germany. 
   \textit{E-mail address}: \texttt{affolter@posteo.net}},
 Jan Techter\thanks{Institut für Mathematik, TU Berlin, Germany. 
    \textit{E-mail address}: \texttt{techter@tu-berlin.de}} 
    \bigskip\bigskip\\
}

\date{\today}

\begin{document}

\maketitle

\noindent
\textbf{Abstract.}
Principal binets are a discretization of curvature line parametrized surfaces
defined on the vertices and faces of the square lattice $\Z^2$.
They generalize the previously established discretizations given by circular nets, conical nets, and principal contact element nets.
We show that principal binets constitute a discrete integrable system in the sense of multi-dimensional consistency.
In particular, they generalize to higher-dimensional square lattices $\Z^N$.
We also discuss relations to the notion of discrete orthogonal coordinate systems as previously established for discrete confocal quadrics.

\newpage

\setcounter{tocdepth}{2}
\tableofcontents

\newpage

\section{Introduction} \label{sec:introduction}

\begin{figure}[t]
  \centering
  \begin{overpic}[width=0.6\textwidth]{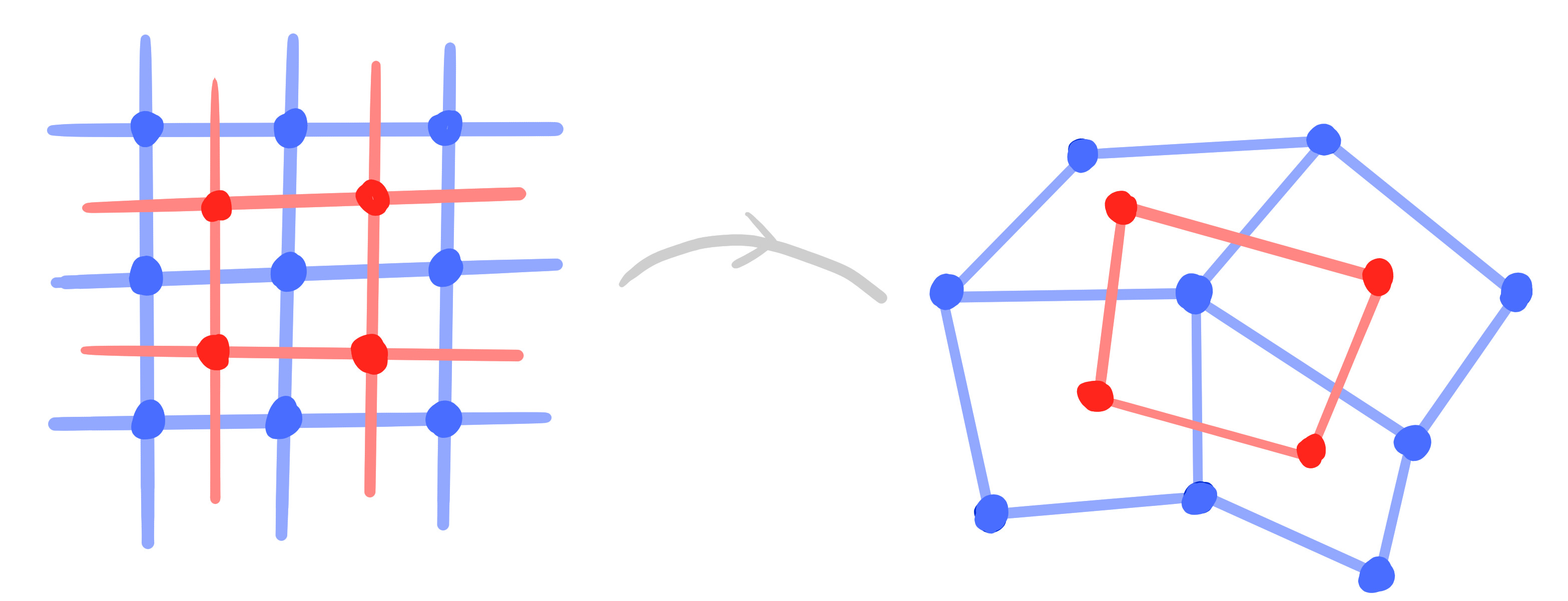}
    \put(3,5){$D_2$}
    \put(90,30){$\RP^n$}
  \end{overpic}
  \caption{A (two-dimensional) binet is a map defined on both the vertices and faces of  $\Z^2$.}
  \label{fig:binet}
\end{figure}

Discrete differential geometry translates the concepts of smooth differential geometry into a discrete setting,
replacing smooth surfaces with discrete meshes that retain their fundamental properties.
A common approach is to discretize parametrized surfaces as maps defined on $\Z^2$, known as \emph{discrete nets} \cite{ddgbook}.
Of special interest are specific types of parametrizations.
For example, conjugate parametrizations have been discretized as \emph{discrete conjugate nets} \cite{sauerqnet}.
Principal (curvature line) parametrizations, in turn, have been discretized as \emph{circular nets} \cite{bpisothermic, bobenkocircular}, \emph{conical nets} \cite{lpwywconical}, and \emph{principal contact element nets} \cite{bsorganizing}.
Recently, \emph{binets} have been introduced as a more general approach to discretizing parametrized surfaces \cite{atbinets}.
These are defined as maps on both the vertices \emph{and} the faces of $\Z^2$ (see Figure~\ref{fig:binet}).
Notably, the special case of \emph{principal binets} led to a generalization of the aforementioned previously established discretizations of principal parametrizations.

In their seminal work \cite{bsorganizing}, Bobenko and Suris formulated two main principles that structure preserving discretizations should adhere to,
namely the \emph{transformation group principle} and the \emph{consistency principle}.
The latter principle asserts that discretizations should inherit
the integrability properties of their smooth counterparts.  
This idea is formalized by the notion of
\emph{multi-dimensional consistency}, which concerns the ability to extend discrete surface theories consistently from maps on $\Z^2$ to maps on $\Z^N$.
While \cite{atbinets} established the transformation group principle for principal binets,
the present paper focuses on the consistency principle.
We generalize binets from $\Z^2$ to $\Z^N$, defining them as maps on the vertices and (2-dimensional) faces of $\Z^N$.
Our main result is the multi-dimensional consistency of principal binets.
In the process, we also establish the consistency of \emph{conjugate binets} and \emph{polar conjugate binets}.

In view of the smooth theory, the additional directions in a discrete net on $\Z^N$ can be interpreted in two ways: as new parameters in a higher-dimensional parametrization, or as  fundamental transformations of these parametrizations.
For principal binets, these two perspectives allow an interpretation as either \emph{discrete orthogonal coordinate systems} or as their \emph{Ribaucour transformations}.
We investigate both of these viewpoints,
and in particular discuss relations to the notion of discrete orthogonal coordinate systems as previously established for discrete confocal quadrics \cite{bsstconfocali, bsstconfocalii}.

\subsection*{Results and context}

The most prominent classical discretization of conjugate parametrizations on smooth surfaces are \emph{discrete conjugate nets},
characterized by the property that all elementary quadrilaterals are planar \cite{sauerqnet}.
We introduce three types of conjugate nets on $\Z^N$ as follows.

Let $V_N, E_N, F_N$ denote the vertices, edges, and (2-dimensional) faces of $\Z^N$, respectively.  
Moreover, let $D_N = V_N \cup F_N$ denote the union of vertices and faces of $\Z^N$.  
We say that two vertices $v, v' \in V_N$ (resp. two faces $f, f' \in F_N$) are \emph{adjacent} if they share an edge, i.e., $(v, v') \in E_N$.  
A vertex $v \in V_N$ and a face $f \in F_N$ are called \emph{incident} if $v$ is a boundary vertex of $f$, and we write $v \sim f$.
We denote by $\RP^n$ the $n$-dimensional real projective space.

\begin{figure}[tb]
  \centering
  \raisebox{0.15\textwidth}{
    \begin{overpic}[width=0.19\textwidth]{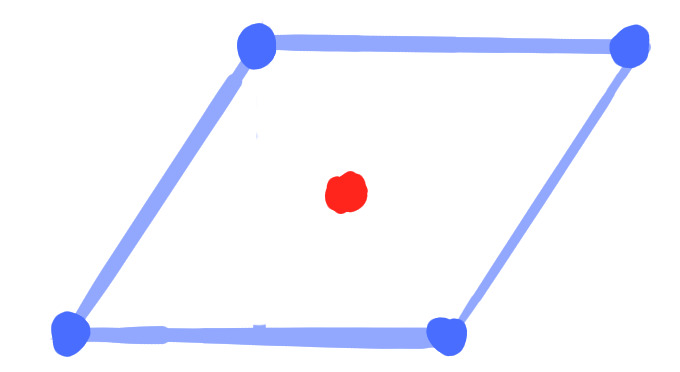}
      \put(53,22){$\color{red}f$}
    \end{overpic}
  }
  \begin{overpic}[width=0.38\textwidth]{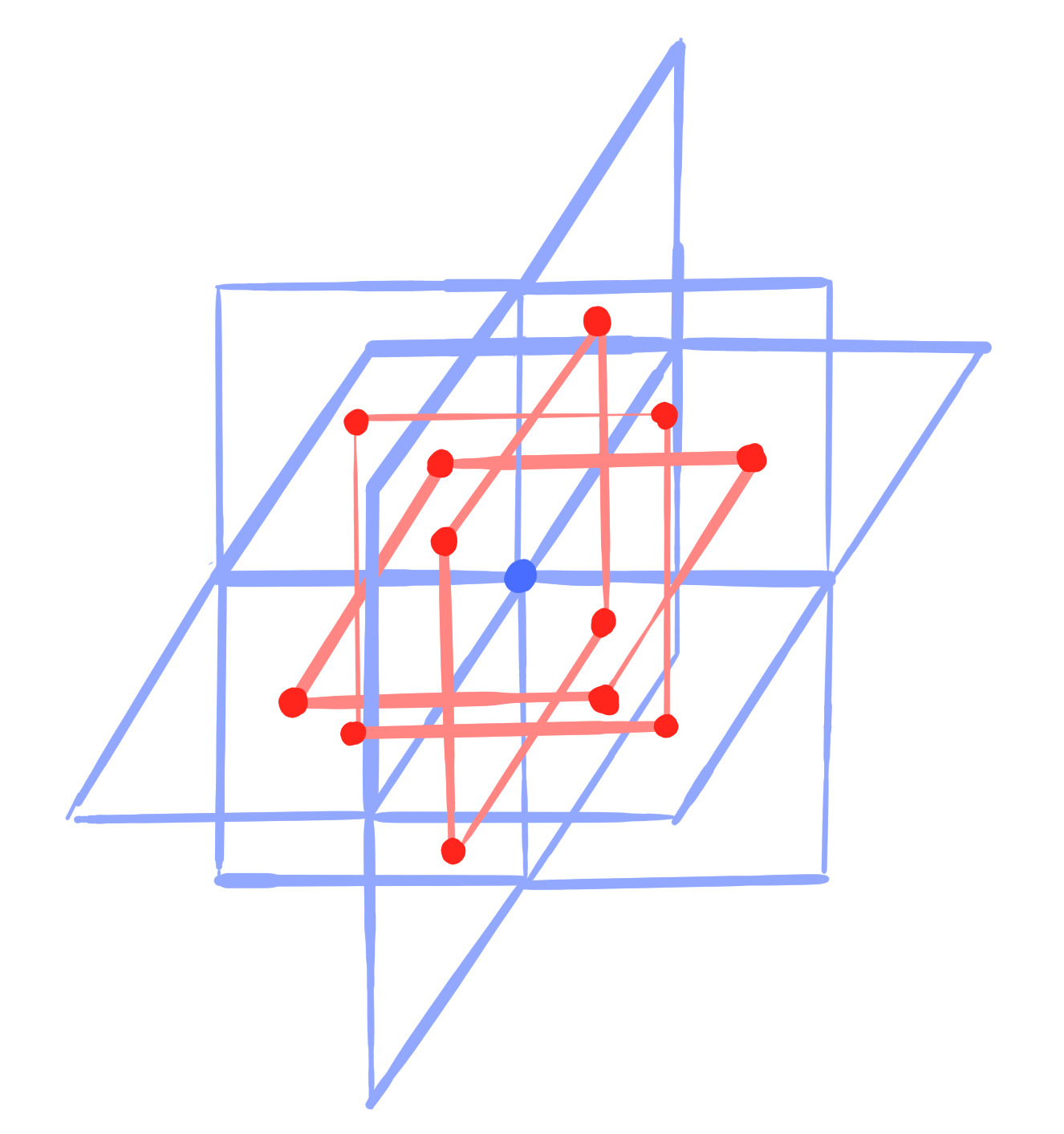}
    \put(49,50){$\color{blue}v$}
  \end{overpic}
  \caption{
    Combinatorics for the conditions of conjugate vertex-nets and conjugate face-nets in $\Z^3$.
    \emph{Left}: A face $f \in F_3$ has four incident vertices.
    The corresponding points of a conjugate vertex-net are coplanar.
    \emph{Right}: A vertex $v \in V_3$ has twelve incident faces.
    The corresponding points of a conjugate face-net are coplanar.
  }
  \label{fig:conjugate-face-nets}
\end{figure}
\begin{samepage}
\begin{definition}\label{def:conjugate}
  We define three types of conjugate nets as follows (cf.\ Figure~\ref{fig:conjugate-face-nets}).
  {
    \setlength{\abovedisplayskip}{6pt}
    \setlength{\belowdisplayskip}{6pt}
    \setlength{\abovedisplayshortskip}{4pt}
    \setlength{\belowdisplayshortskip}{4pt}
    \begin{enumerate}[noitemsep]
    \item
      A \emph{conjugate vertex-net} is a map $g: V_N \rightarrow \RP^n$ such that the span
      \begin{equation}
        \bigvee_{v \sim f} g(v) \quad
        \text{is a plane for every}~f\in F_N.
      \end{equation}
    \item
      \label{itm:conjugateface}
      A \emph{conjugate face-net} is a map $h: F_N \rightarrow \RP^n$ such that the span
      \begin{equation}
        \bigvee_{f \sim v} h(f) \quad
        \text{is a plane for every}~v\in V_N.
      \end{equation}
    \item
      A \emph{conjugate binet} is a map $b: D_N \rightarrow \RP^n$ such that the restriction to $V_N$ is a conjugate vertex-net and the restriction to $F_N$ is a conjugate face-net.
      \qedhere
    \end{enumerate}
  }
\end{definition}
\end{samepage}

Section~\ref{sec:conjugate-binets-preliminaries} discusses non-degeneracy and genericity conditions for all three types of conjugate nets.

The definition of conjugate vertex-nets coincides with that of \emph{quadrilateral lattices} in \cite{dsqnet} and \emph{Q-nets} in \cite{ddgbook}.
The definition of conjugate face-nets can be regarded as a natural higher-dimensional generalization of \emph{Q*-nets}, which were defined only on $\Z^2$ in \cite{ddgbook}.
For $N=2$, our definition of conjugate binets agrees with that in \cite{atbinets}.
In this case, an isomorphism exists that interchanges the roles of vertices $V_2$ and faces $F_2$ of $\Z^2$,
allowing a conjugate face-net on $F_2$ to be identified with a conjugate vertex-net.
Consequently, a conjugate binet on $D_2$ may be viewed as a pair of conjugate vertex-nets.
For $N>2$, however, no such bijection between $V_N$ and $F_N$ exists.
Therefore, conjugate face-nets are generally distinct from conjugate vertex-nets,
and thus we use the terms ``vertex-net'' and ``face-net'' to clearly distinguish between these combinatorially different cases in this paper.

However, we show that there is still a close correspondence between conjugate face-nets and conjugate vertex-nets.  
For $1 \leq i < j \leq N$, let $F^{ij}_N \subset F_N$ denote the (2-dimensional) $ij$-faces of $\Z^N$, that is, the faces spanned by the $i$-th and $j$-th unit vectors.  
Clearly, there is a natural bijection between $V_N$ and $F_N^{ij}$ for any $i,j$ (cf.\ Figure~\ref{fig:face-identification}).  
Thus, we may consider conjugate vertex-nets on $F_N^{ij}$.  
There is a one-to-one correspondence between conjugate face-nets on $F_N$ and conjugate vertex-nets on $F_N^{ij}$, which is established in Section~\ref{sec:conjugate-binets-correspondence} and formalized in the following theorem.  

\begin{figure}[tb]
  \centering
  \includegraphics[width=0.35\textwidth]{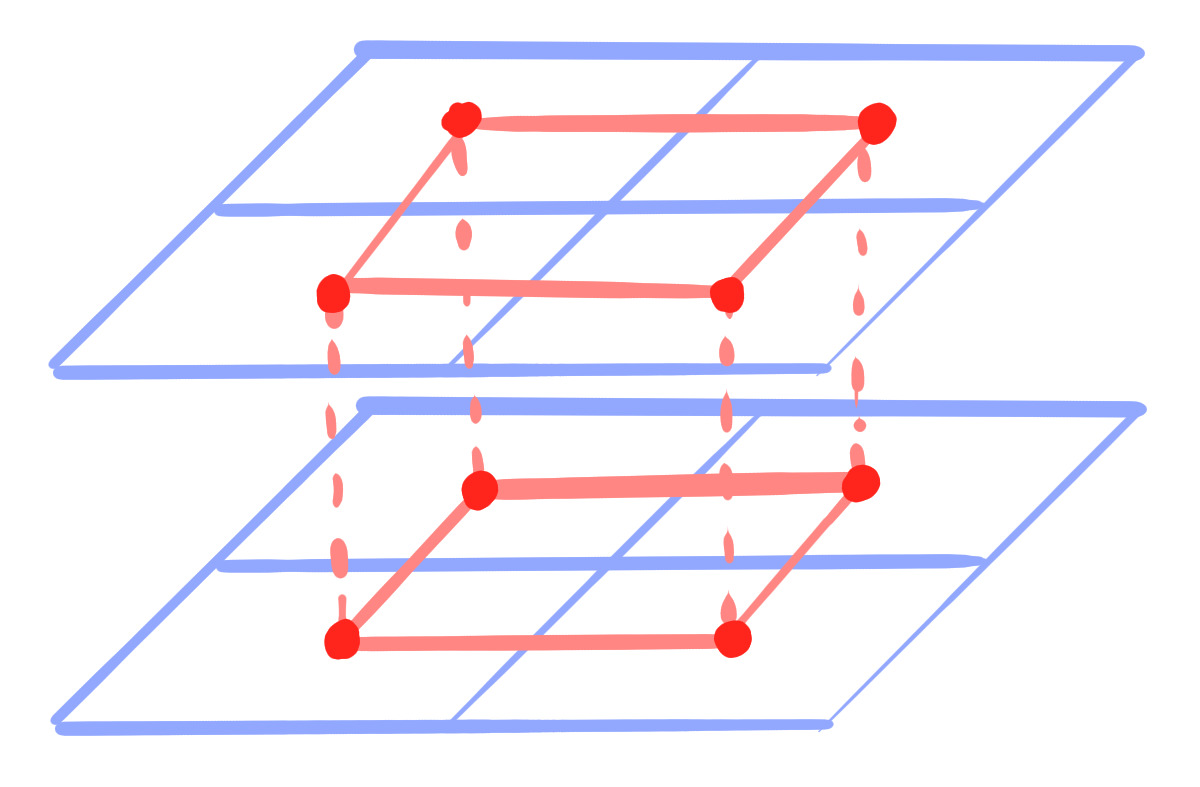}
  \caption{
    Identification of $F_3^{12}$  with $\Z^3$.
    In blue two layers of $12$-faces of $\Z^3$.
    In red $F_3^{12} \simeq \Z^3$.
    The dashed edges are the missing edges for the identification with $\Z^3$.
  }
  \label{fig:face-identification}
\end{figure}

\begin{theorem}\label{th:conjugatebijection}
  Let $1 \leq i < j \leq N$. 
  For every conjugate face-net defined on $F_N$ the restriction to $F_N^{ij}$ is a conjugate vertex-net.
  Conversely, every conjugate vertex-net defined on $F_N^{ij}$ is the restriction of a unique conjugate face-net defined on $F_N$.
\end{theorem}

Hence, to some extent, we may think of a conjugate binet as a pair of conjugate vertex-nets,
even for $N > 2$, up to the choice of a pair of directions $i,j$.  
Building on this correspondence, we extend the result on multi-dimensional consistency of conjugate vertex-nets from \cite{dsqnet} to conjugate face-nets, and thus conjugate binets, as summarized in the following theorem and established in Section~\ref{sec:conjugate-binets-consistency}.  

\begin{theorem}
  \label{th:conjugateconsistency}
  Conjugate vertex-nets, conjugate face-nets and conjugate binets are multi-di\-men\-sionally consistent 3D-systems.
\end{theorem}

The first classical discretization of principal (curvature line) parametrizations on smooth surfaces
are \emph{circular nets} \cite{bpisothermic, bobenkocircular},
which form a special class of conjugate vertex-nets.
Their multi-dimensional consistency (as a reduction of conjugate vertex-nets) was established in \cite{cdscircular} using Miquel's theorem.
Alternatively, one can observe that circular nets are the stereographic projections of conjugate vertex-nets inscribed in the Möbius quadric (a sphere in $\RP^{n+1}$) \cite{ddgbook}.
The result then follows from the more general theorem that conjugate vertex-nets in quadrics
are consistent reductions of conjugate vertex-nets \cite{doliwaqnetsquadrics}.

A novel discretization of principal parametrizations put forward in \cite{atbinets} are \emph{principal binets}.
These are conjugate binets with orthogonal pairs of dual edges.
Independently, and from a different perspective, a closely related discretization was introduced in \cite{dellingercbpatterns} under the name
\emph{principal checkerboard patterns}.
Similar to circular nets, principal binets can be realized as central projections of conjugate vertex-nets in $\RP^{n+1}$, called \emph{Möbius lifts}, in which incident vertices and faces are related by polarity with respect to the Möbius quadric.
Möbius lifts of principal binets are special cases of \emph{polar conjugate binets}, which generalize conjugate vertex-nets in quadrics.
In analogy with the relation between circular nets and conjugate vertex-nets in quadrics,
our proof of the multi-dimensional consistency of principal binets
is based on the multi-dimensional consistency of polar conjugate binets.

\begin{definition}
  \label{def:polar-conjugate}
  Let $\mathcal Q$ be a quadric in $\RP^n$.
  A \emph{polar conjugate binet} $b: D_N \rightarrow \RP^n$
  is a conjugate binet such that for every pair of incident $v \in V_N$ and $f \in F_N$
  the points $b(v)$ and $b(f)$ are related by polarity with respect to $\mathcal Q$.
\end{definition}

For $N=2$ the definition of conjugate polar binets coincides with that given in \cite{atbinets}. 
In Section~\ref{sec:polarbinets}, we prove the following theorem, which generalizes the corresponding result for conjugate vertex-nets in quadrics \cite{doliwaqnetsquadrics}.

\begin{theorem}
  \label{th:polarconsistency}
  Polar conjugate binets are a consistent reduction of conjugate binets.
\end{theorem}

\begin{figure}[t]
  \centering
  \includegraphics[width=0.3\textwidth, trim={0 90 0 20}, clip]{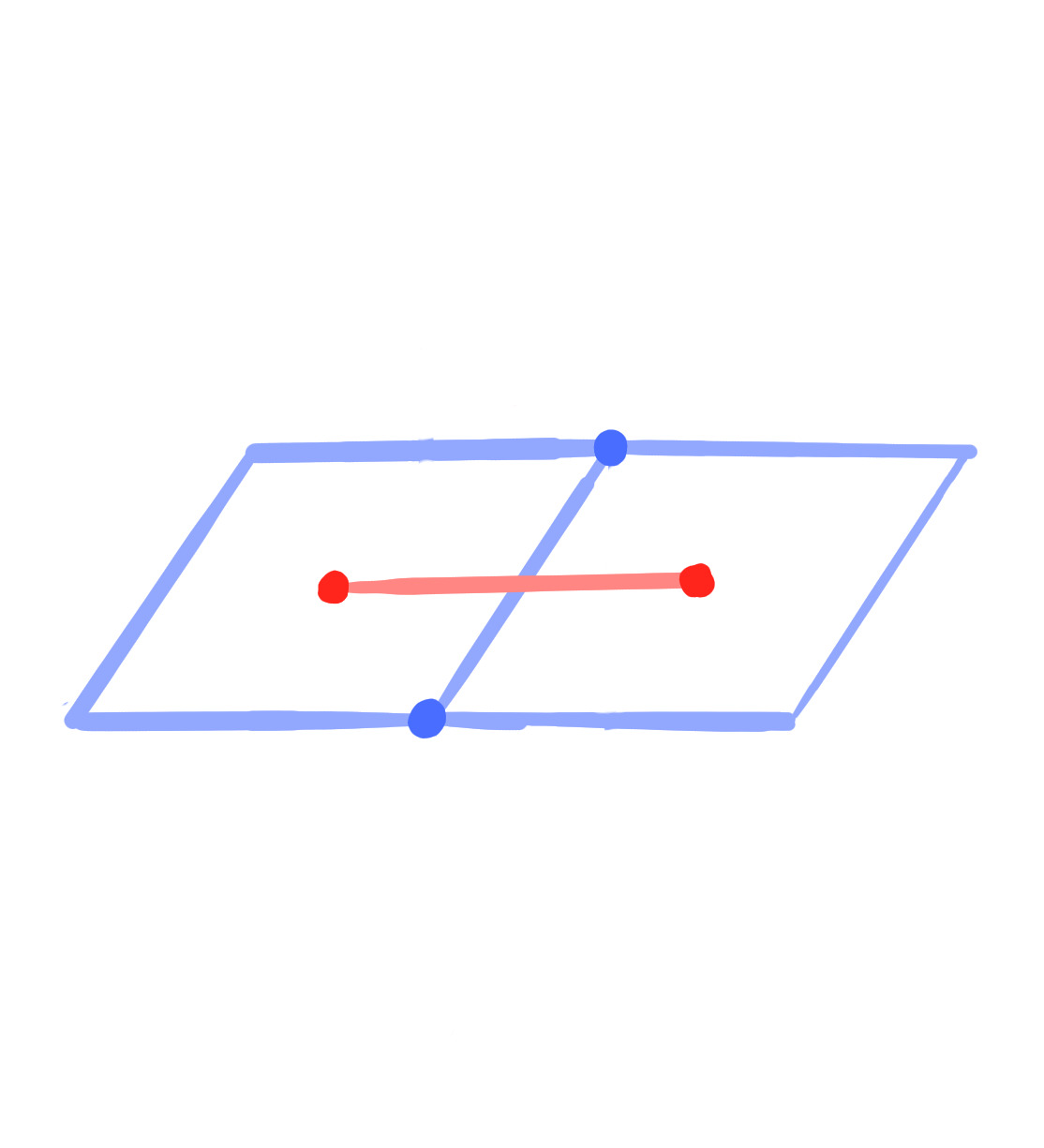}
  \includegraphics[width=0.3\textwidth, trim={0 90 0 20}, clip]{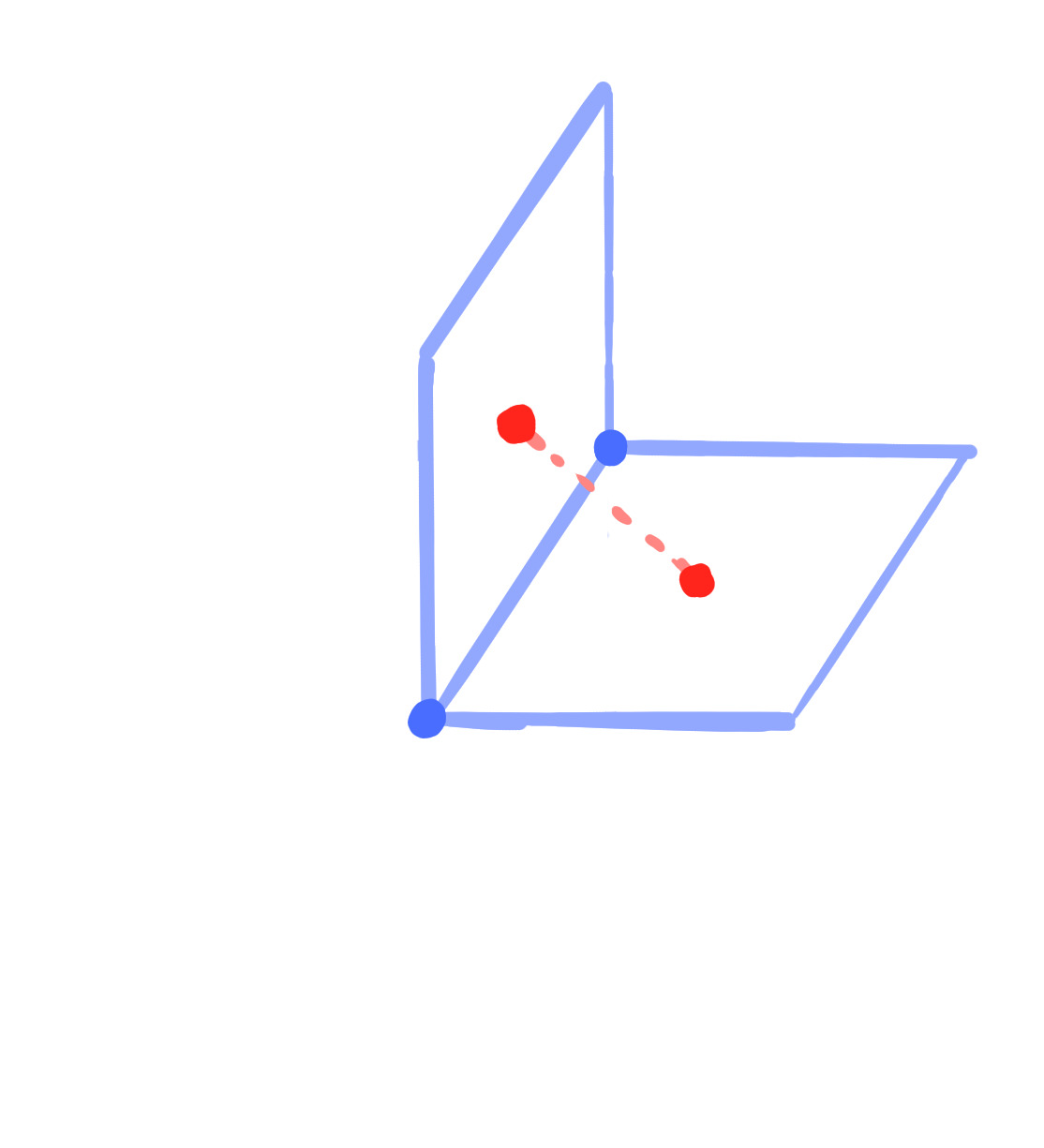}
  \caption{
    The two types of crosses in $\Z^3$.
    This amounts to a total of six crosses per edge of $\Z^3$.}
  \label{fig:crosses-types}
\end{figure}

Let $C_N$ denote the set of \emph{crosses} of $\Z^N$, defined as tuples consisting of two vertices and two faces such that both vertices are incident to both faces (cf.\ Figure~\ref{fig:crosses-types}):
\begin{equation}
  C_N \coloneqq \set{(v,f,v',f')}{v, v' \in V_N, \ f, f' \in F_N, \ v, v' \sim f, f'}.
\end{equation}

\begin{definition}
  \label{def:principal-binet}
  A \emph{principal binet} $b: D_N \rightarrow \R^n$ is a conjugate binet
  such that for every cross $(v,f,v',f') \in C_N$ the two lines $b(v) \vee b(v')$ and $b(f) \vee b(f')$ are orthogonal.
\end{definition}

For $N=2$, the definition of principal binets coincides with that given in \cite{atbinets}. 
Our main result is the following.
It may be viewed as a generalization of the corresponding result for circular nets \cite{cdscircular}.

\begin{theorem}\label{th:principalconsistency}
    Principal binets are a consistent reduction of conjugate binets.
\end{theorem}

The proof of this theorem is given in Section~\ref{sec:principal-binets},
and is based on a $1\!:\!\R$ correspondence between principal binets and their Möbius lifts,
which are polar conjugate binets.

A pair of conjugate vertex-nets can be extended to a conjugate binet in two ways,
depending on which of the two nets is viewed as being defined on $V_N$ or on $F_N^{ij}$, respectively.
Principal binets exhibit a surprising symmetry with respect to this interchange,
which we record in Theorem~\ref{th:principalsymmetry}.

Circular nets constitute a special case of principal binets \cite{atbinets}.
This remains true in arbitrary dimension, and it turns out that principal binets exhibit similar geometric structures to circular nets, that allow for interpretation in terms of \emph{Ribaucour transformations} (Section~\ref{sec:circular-nets}).

As laid out in \cite{atbinets}, principal binets can be viewed not only as a generalization of circular nets,
but, in the case $n=3$, also of \emph{conical nets} \cite{lpwywconical}, another discretization of principal parametrizations.
In particular, analogously to circular nets, which are a consistent reduction of conjugate vertex-nets, conical nets may be regarded as a consistent reduction of conjugate face-nets, as stated in Theorem~\ref{thm:conical-nets-as-face-nets} (Section~\ref{sec:conical-nets}).

A unifying perspective on circular and conical nets has been presented in terms of \emph{principal contact element nets} \cite{bsorganizing},
which correspond to \emph{isotropic line complexes} in the Lie quadric.
The Lie geometric description of principal binets naturally gives rise to the more general notion of \emph{polar line bicomplexes} \cite{atbinets}.
We will use known results on line complexes to establish the multi-dimensional consistency of polar line bicomplexes, providing an alternative proof of the multi-dimensional consistency of principal binets, at least in the case $n=3$ (Section~\ref{sec:principal-contact-element-nets}).

One of the motivating examples for the introduction of principal binets in \cite{atbinets} are the
\emph{discrete confocal quadrics}, introduced in \cite{bsstconfocali, bsstconfocalii}
(see Figure~\ref{fig:discrete-confocal}, left).
However, for $N>2$, the combinatorics underlying discrete confocal quadrics differs from those of principal binets.
In particular, for $N=3$, the points of a \emph{discrete orthogonal coordinate system} used in the construction of
discrete confocal quadrics are defined on the vertices and elementary cubes of $\Z^3$,
rather than on the vertices and faces, as is the case for principal binets.
In Section~\ref{sec:toc}, we establish a direct correspondence between principal binets
and discrete orthogonal coordinate systems, which is formalized in Theorem~\ref{thm:binets-docs}.

Extending the definition of a principal binet from $\Z^2$ by a third coordinate direction allows us to view it as either a sequence of Ribaucour transformations of a principal parametrization or as an extension to an orthogonal coordinate system.
In the latter interpretation, and when compared with the smooth theory, we conclude that a natural interpretation is to regard only a distinguished subset of points of a principal binet as discretizations of points of a smooth orthogonal coordinate system.
The remaining points can then be viewed as discretizations of certain focal points of that system.
In Section~\ref{sec:smooth-vs-discrete}, we show that the orthogonality conditions imposed on principal binets correspond to those of a smooth orthogonal coordinate system.
As it turns out, also in the smooth setting, an orthogonal coordinate system can be characterized by imposing suitable orthogonality conditions on the focal points of a triply conjugate system (Theorem~\ref{thm:tos-from-focal-points}).

\subsection*{Acknowledgments}
Both authors were supported by the Deutsche Forschungsgemeinschaft (DFG) Collaborative Research Center TRR 109 ``Discretization in Geometry and Dynamics''.
The images in this paper were created with \texttt{Krita}, \texttt{GeoGebra}, and \texttt{Blender}.

\section{Conjugate binets}
\label{sec:conjugate-binets}

\subsection{Technical Preliminaries}
\label{sec:conjugate-binets-preliminaries}

We have introduced \emph{conjugate vertex-nets}, \emph{conjugate face-nets}, and \emph{conjugate binets} in Definition~\ref{def:conjugate} as maps to $\RP^n$.
Throughout the paper, we assume that $n \geq 3$, so that the constraints defining the three kinds of conjugate nets are meaningful.

We further assume that all conjugate nets are \emph{non-degenerate} in the following sense.
\begin{definition}\
  \label{def:non-degeneracy}
  \nobreakpar
  \begin{enumerate}
  \item
    \label{def:non-degeneracy-vertex-net}
    A conjugate vertex-net is \emph{non-degenerate} if the image of any three vertices of a quad spans a plane.
  \item
    \label{def:non-degeneracy-face-net}
    A conjugate face-net is \emph{non-degenerate} if the image of any three faces incident with a vertex spans a plane,
    and if the planes associated in this way with any three vertices of a quad intersect in exactly one point.
    \qedhere
  \end{enumerate}
\end{definition}

Since a conjugate binet is the union of a conjugate vertex-net and a conjugate face-net,
we may apply these and all subsequent notions introduced in this section directly to conjugate binets as well.
Under the assumption of non-degeneracy, we introduce the following notation
for the planes occurring in conjugate nets.
\begin{definition}\
  \nobreakpar
  \label{def:planes}
  \begin{enumerate}
  \item
    Let $g : V_N \rightarrow \RP^n$ be a conjugate vertex-net.
    Then for every face $f \in F_N$ we denote the plane spanned by the images of the incident vertices by
    \begin{align}
      \square g(f) \coloneqq \bigvee_{v \sim f} g(v).
    \end{align}
  \item
    \label{def:planes-face-net}
    Let $h : F_N \rightarrow \RP^n$ be a conjugate face-net.
    Then for every vertex $v \in V_N$ we denote the plane spanned by the images of the incident faces by
    \begin{align}
      \square h(v) \coloneqq \bigvee_{f \sim v} h(f).  
    \end{align}
      \qedhere
  \end{enumerate}
\end{definition}

\begin{remark}
  \label{rem:conjugate-face-net-planes}
  Instead of the points of a conjugate face-net,
  one may regard the planes defined on vertices as the primary objects.
  This leads to a map
  \begin{equation}
    \square h : V_N \rightarrow \mathrm{Planes}(\RP^n),
  \end{equation}
  introduced as a \emph{conjugate *net} in \cite{atbinets}.
  In this formulation, the non-degeneracy condition in
  Definition~\ref{def:non-degeneracy}~\ref{def:non-degeneracy-face-net}
  reduces to the requirement that the planes associated with any three vertices of a quad
  intersect in exactly one point, so that the conditions
  \ref{def:non-degeneracy-vertex-net} and \ref{def:non-degeneracy-face-net}
  become fully symmetric.
  Moreover, in this description, a conjugate binet is a pair consisting of a conjugate vertex-net
  and a conjugate *net, both defined on the vertices of $\Z^N$.
\end{remark}

To ensure that our proofs apply to $\RP^n$ with $n \geq 3$,
we require conjugate nets to satisfy certain additional genericity conditions.
When these conditions are not satisfied by a given net,
they are still satisfied by a \emph{lift} to a higher-dimensional projective space,
and we will see that our proofs are invariant under replacing a conjugate net by such a lift.
\begin{definition}\
  \nobreakpar
  \begin{enumerate}
  \item
    Let $g : V_N \rightarrow \RP^n$ be a conjugate vertex-net.
    A map $\hat g : V_N \rightarrow \RP^{n+m}$ is called a \emph{lift} of $g$
    if $\hat g$ is a conjugate vertex-net and there exists a central projection
    $\pi : \RP^{n+m} \setminus C \rightarrow \RP^n$ with center $C$ such that
    \begin{equation}
      \pi \circ \hat g = g,
    \end{equation}
    where $\RP^n$ is embedded into $\RP^{n+m}$ as a projective subspace,
    $C \subset \RP^{n+m}$ is a projective subspace of dimension $m-1$ disjoint from $\RP^n$,
    and $\pi$ is defined by
    \begin{equation}
      \pi : X \mapsto (X \vee C) \cap \RP^n.
    \end{equation}
  \item
    The lift of a conjugate face-net is defined analogously.
    \qedhere
  \end{enumerate}
\end{definition}

\begin{definition}
  We call a subset $A \subset \Z^N$ a \emph{block} if it is of the form
  \begin{align}
    A = \Z^N \cap \bigotimes_{i=1}^N [0, a_i],
  \end{align}
  for some $a_1, \dots, a_N \in \Z$.
  We denote the sum of side lengths of a block $A$ by $\delta(A) \coloneqq \sum_{i=1}^N a_i$.
  \qedhere
\end{definition}

\begin{lemma}
  \label{lem:dimvertexnets}
  Let $g : A \rightarrow \RP^n$ be the restriction of a conjugate vertex-net to a block $A \subset V_N$.
  Then $g$ admits a lift $\hat g$ such that the dimension spanned by the image of $\hat g$
  on every block $B \subset A$ is equal to $\delta(B)$.
\end{lemma}
\begin{proof}
  By the non-degeneracy assumption on conjugate nets,
  the span of the images of two adjacent points is one-dimensional.
  Using this, we may choose $\hat g$ such that the stated condition holds
  for all subsets $B \subset A$ with all but one side length equal to zero,
  i.e., along the coordinate axes of $\Z^N$.
  This implies that the image of $\hat g$ on any block $B \subset A$
  spans a space of at least dimension $\delta(B)$.
  Since all remaining points are contained in the span of the coordinate axes,
  equality follows.
\end{proof}

\begin{lemma}
  \label{lem:dimfacenets}
  Let $\square h : A \rightarrow \mathrm{Planes}(\RP^n)$ be the restriction of the planes
  of a conjugate face-net to a block $A \subset V_N$.
  Then $\square h$ admits a lift $\square \hat h$ such that the dimension spanned by the image
  of $\square \hat h$ on every block $B \subset A$ is equal to $2 + \delta(B)$.
\end{lemma}

\begin{proof}
  The proof is completely analogous to that of Lemma~\ref{lem:dimvertexnets},
  except that the non-degeneracy assumption implies that the span of the images
  of two adjacent planes is three-dimensional.
\end{proof}

\subsection{Correspondence between conjugate face-nets and conjugate vertex-nets}
\label{sec:conjugate-binets-correspondence}

As explained briefly in the introduction,
while the definitions of conjugate vertex-nets and conjugate face-nets
appear to be very symmetric with respect to $V_N \leftrightarrow F_N$,
the set of vertices and faces are not isomorphic for $N>2$.
Therefore, conjugate face-nets are a priori different from conjugate vertex-nets.
For $1 \leq i < j \leq N$, recall that $F_N^{ij}$ denotes the (2-dimensional) $ij$-faces of $\Z^N$.
We may identify $F_N^{ij}$ with $V_N$ by
\begin{align}
	F_N^{ij} \ni f = (r, r+e_i, r + e_i +e_j, r+e_j) \longleftrightarrow r \in V_N.
\end{align}
Note that this identification introduces edges between $ij$-faces that were previously not adjacent in $\Z^N$ (cf.\ Figure~\ref{fig:face-identification}).
For a map $h$ defined on $F_N$, we denote by $h^{ij}$ its restriction to the $ij$-faces.

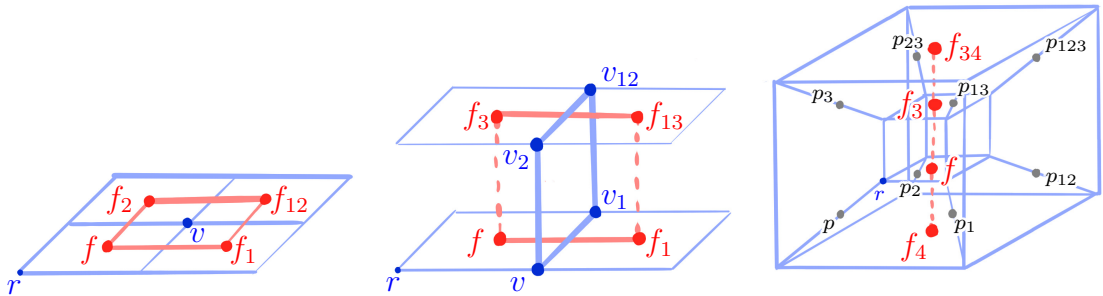
\begin{figure}[b]
  \centering
  \begin{overpic}[width=0.3\textwidth]{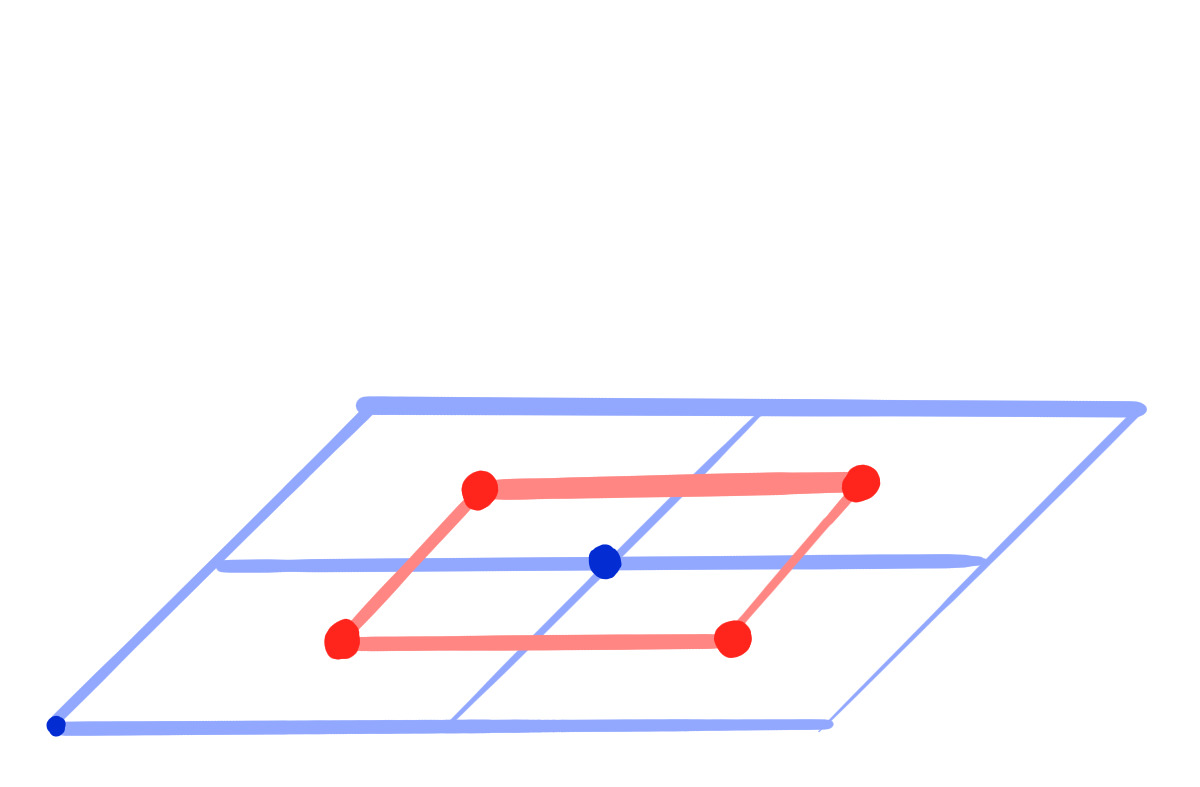}
    \put(1,-1){$\color{blue}r$}
    \put(51,13.5){$\color{blue}v$}
    \put(21,9){$\color{red}f$}
    \put(62,9){$\color{red}f_1$}
    \put(73,23){$\color{red}f_{12}$}
    \put(28,23){$\color{red}f_2$}
  \end{overpic}
  \begin{overpic}[width=0.3\textwidth]{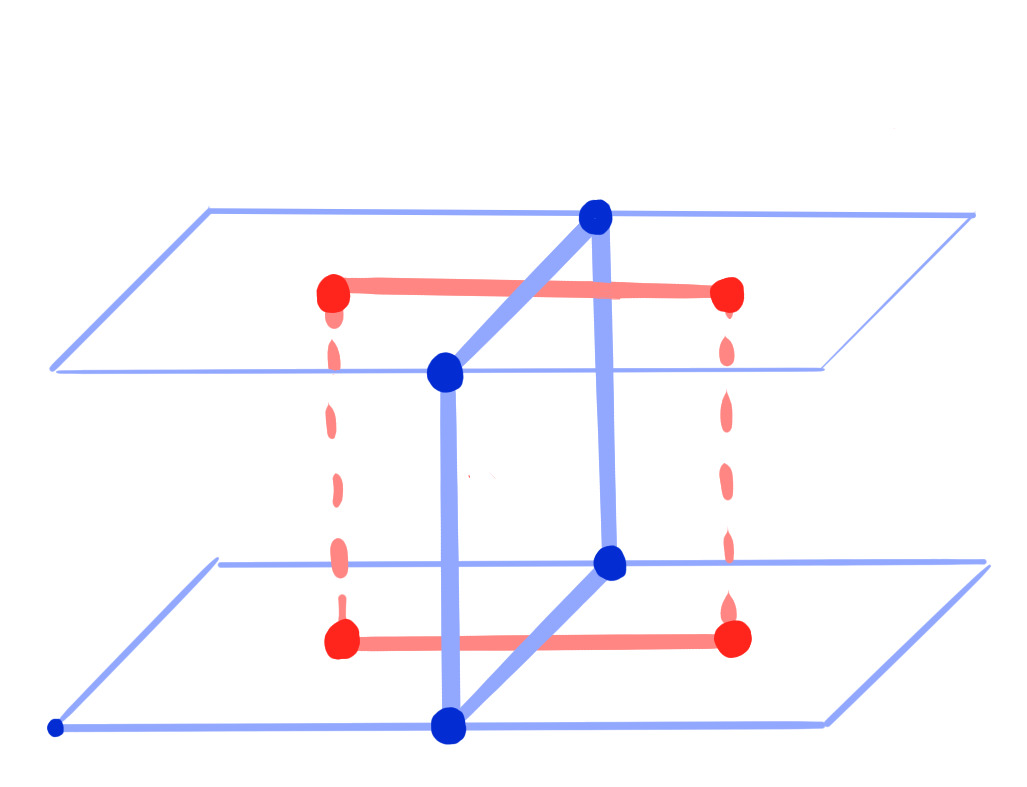}
    \put(2,0){$\color{blue}r$}
    \put(36,0){$\color{blue}v$}
    \put(61,24){$\color{blue}v_1$}
    \put(61,58){$\color{blue}v_{12}$}
    \put(34,35){$\color{blue}v_2$}
    \put(24,11){$\color{red}f$}
    \put(73,11){$\color{red}f_1$}
    \put(73,46){$\color{red}f_{13}$}
    \put(23,46){$\color{red}f_3$}
  \end{overpic}
  \begin{overpic}[width=0.3\textwidth]{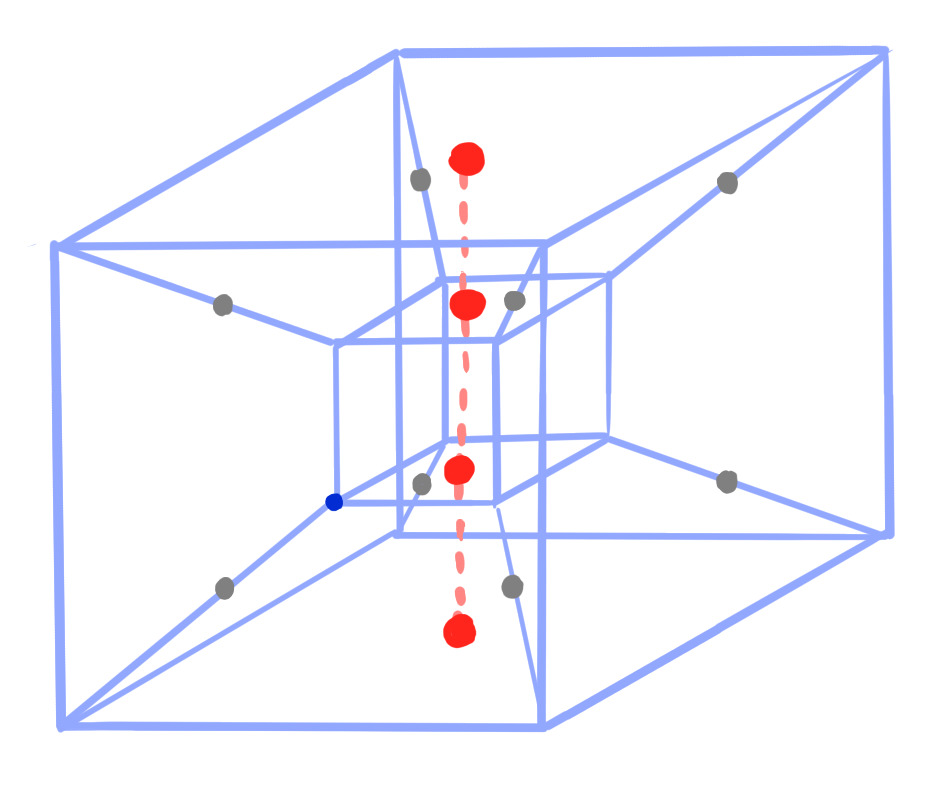}
    \put(33,25){\scriptsize\contour{white}{$\color{blue}r$}}
    \put(40,10){$\color{red}f_4$}
    \put(51,31){\contour{white}{$\color{red}f$}}
    \put(39,49){\contour{white}{$\color{red}f_3$}}
    \put(52,65){$\color{red}f_{34}$}
    \put(19,17){\scriptsize\contour{white}{$p$}}
    \put(55,17){\scriptsize\contour{white}{$p_1$}}
    \put(16,53){\scriptsize\contour{white}{$p_3$}}
    \put(56,55){\scriptsize\contour{white}{$p_{13}$}}
    \put(40,27){\scriptsize\contour{white}{$p_2$}}
    \put(80,30){\scriptsize\contour{white}{$p_{12}$}}
    \put(80,67){\scriptsize\contour{white}{$p_{123}$}}
    \put(38,68){\scriptsize\contour{white}{$p_{23}$}}
  \end{overpic}
  \caption{Different face types of $F_N^{12} \cong V_N$.
    Labels in the figure correspond to the notation used in the proof of Proposition~\ref{prop:conjugatefacetoconjugate}.
    \emph{Left}: $12$-face of $F_N^{12}$.
    \emph{Middle}: $13$-face of $F_N^{12}$ with an associated $23$-face of $\Z^N$.
    \emph{Right}: $34$-face of $F_N^{12}$ with a surrounding 4-cube of $\Z^N$.}
  \label{fig:conjugate-face-net-restriction}
\end{figure}

\begin{proposition}
  \label{prop:conjugatefacetoconjugate}
    The restriction $h^{ij}: F_N^{ij} \rightarrow \RP^n$  of a conjugate face-net $h: F_N \rightarrow \RP^n$ to $F_N^{ij} \simeq V_N$ is a conjugate vertex-net.
  \end{proposition}
\begin{proof}
  Without loss of generality, let $i=1, j=2$.
  We verify planarity for all face types of $F_N^{12}$.
  \begin{itemize}
  \item
    $12$-faces:\\
    Clearly, the $12$-faces of $h^{12}$ are planar.
    Indeed, consider a $12$-face $(f, f_1, f_{12}, f_2)$ of $F_N^{12}$ with incident vertex $v \in V_N$
    (see Figure~\ref{fig:conjugate-face-net-restriction}, left), i.e.,
    \begin{equation}
      r + e_1 + e_2 = v \ \sim\ f,\, f_1,\, f_{12},\, f_2 \ \longleftrightarrow\  r,\, r+e_1,\, r+e_1+e_2,\, r+e_2
    \end{equation}
    with some $r \in V_N$.
    Then, since $h$ is a conjugate vertex-net, the four points lie in the plane $\square h(v)$,
    \begin{align}
      h^{12}(f),\, h^{12}(f_1),\, h^{12}(f_{12}),\, h^{12}(f_2) \in \square h(v).
    \end{align}
  \item
    $13$-faces (analogously $kl$-faces with $k=1,2$ and $l \geq 3$):\\
    Consider a $13$-face $(f, f_1, f_{13}, f_2)$ of $F_N^{12}$ and ``incident'' $23$-face $\tilde f = (v, v_2, v_{23}, v_3)$ of $\Z^N$ (see Figure~\ref{fig:conjugate-face-net-restriction}, middle), i.e.,
    \begin{equation}
      \label{eq:13-face}
      \begin{aligned}
      \begin{aligned}
        &f,\, f_1,\, f_{13},\, f_3 &&\longleftrightarrow &&r,\, r+e_1,\, r+e_1+e_3,\, r+e_3\\
        &v,\, v_2,\, v_{23},\, v_3 &&\makebox[\widthof{$\longleftrightarrow$}][c]{=} &&r+e_1,\, r+e_1+e_2,\, r+e_1+e_2+e_3,\, r+e_1+e_3\\        
      \end{aligned}\\
      v, v_2 \sim f, f_1, \qquad v_3, v_{23} \sim f_3, f_{13}, \qquad
      v, v_2, v_{23}, v_3 \sim \tilde f\qquad
    \end{aligned}
    \end{equation}
    with some $r \in V_N$.
    Consider the two lines
    \begin{align}
      L &\coloneqq h^{12}(f) \vee h^{12}(f_1) = \square h(v) \cap \square h(v_2),\\
      L_3 &\coloneqq h^{12}(f_3) \vee h^{12}(f_{13}) = \square h(v_3) \cap \square h(v_{23}).
    \end{align}
    The four planes $\square h(v), \square h(v_2), \square h(v_3), \square h(v_{23})$ intersect in the point $h(\tilde f)$,
    thus the two lines $L$ and $L_3$ are coplanar.
  \item
    $34$-faces (analogously $kl$-faces with $3 \leq k < l$):\\
    Consider a $34$-face $(f, f_3, f_{34}, f_4)$ of $F_N^{12}$, i.e.,
    \begin{align}
      f,\, f_3,\, f_{34},\, f_4 \ \longleftrightarrow\ r,\, r+e_3,\, r+e_3+e_4,\, r+e_4
    \end{align}
    with some $r \in V_N$.
    $f, f_3, f_{34}, f_4$ are four faces of the 4-cube $C$ in $\Z^N$ given by
    \begin{equation}
      \set{r+a \cdot e_1+b \cdot e_2 + c \cdot e_3 + d \cdot e_4}{(a, b, c, d) \in \Z_2^4 = \{0,1\}^4}
    \end{equation}
    (see Figure~\ref{fig:conjugate-face-net-restriction}, right).
    
    In the following, we assume that $h$ satisfies the dimensional statement of Lemma~\ref{lem:dimfacenets}.
    Otherwise, we could replace $h$ with a lift $\hat h$.
    It is straightforward to verify that in what follows any conclusions drawn for $\hat h$ are invariant under the central projection $\pi$ and hence also hold for $h$.
    Thus, we may assume that
    two adjacent planes of $\square h$ span a 3-space,
    the four planes of a quad span a 4-space,
    the eight planes of a cube span a 5-space,
    and the sixteen planes of a 4-cube span a 6-space,
    respectively.

    Consider the eight 3-spaces
    \begin{align}
      A_{abc} &\coloneqq \square h(r+a \cdot e_1+b \cdot e_2 + c \cdot e_3) \vee \square h(r+a \cdot e_1+b \cdot e_2 + c \cdot e_3 + e_4),
    \end{align}
    where $(a, b, c) \in \Z_2^3 = \{0,1\}^3 $.
    Thereby, we associated the eight 3-spaces $A_{abc}$ with the vertices of a 3-cube,
    where each vertex, in turn, may be associated with an edge in 4-direction of the 4-cube $C$.
    
    Let $q^{12} = (p,p_1,p_{12},p_2)$ denote the vertices of a quad in the 3-cube $\Z_2^3$.
    The intersection $A_p \cap A_{p_1}$ is a plane, since the two 3-spaces belong to a quad of the 4-cube $C$.
    Consequently, $A_p \cap A_{p_1} \cap A_{p_2}$ is a line.
    In particular, we have that 
    \begin{align}
      h^{12}(f) \vee h^{12}(f_4) = A_{000} \cap A_{100} \cap A_{010} = A_{000} \cap A_{100} \cap A_{010} \cap A_{110},\\
      h^{12}(f_3) \vee h^{12}(f_{34}) = A_{001} \cap A_{101} \cap A_{011} = A_{001} \cap A_{101} \cap A_{011} \cap A_{111}.
    \end{align}	
    Similarly, for an arbitrary quad $q^{12}$ we obtain that
    \begin{align}
      B_{q^{12}} \coloneqq A_p \cap A_{p_1} \cap A_{p_2} \cap A_{p_{12}},
    \end{align}	
    is a line as well.
    
    Next, let $q^{13} \coloneqq (p,p_1,p_{13},p_3)$ be the vertices of another quad in the 3-cube $\Z_2^3$, which shares an edge with $q^{12}$. Since both lines $B_{q^{12}}$, $B_{q^{13}}$ are contained in the plane $A_p \cap A_{p_1}$, the two lines intersect in a point. Moreover, the intersection 
    \begin{align}
      A_p \cap A_{p_1} \cap A_{p_2} \cap A_{p_3} = \bigcap_{m=1}^3 (A_p \cap A_{p_m}),
    \end{align}
    is a point, since it is the intersection of three planes in the 3-space $A_p$. As a result, the intersection $B_{q^{12}} \cap B_{q^{13}} \cap B_{q^{23}}$ is also a point.
    In other words, the intersections of the lines associated with three pairwise adjacent quads of the 3-cube intersect in a point.
    By iterating this argument, we obtain that the lines associated with any two quads of the 3-cube intersect in a point.
    In particular, we obtain that the intersection
    \begin{align}
      (h^{12}(f) \vee h^{12}(f_4)) \cap (h^{12}(f_3) \vee h^{12}(f_{34}))
    \end{align}
    is a point. This is equivalent to the planarity of the $34$-faces and completes the proof.
    \qedhere
  \end{itemize}
\end{proof}

We next prove the converse of the previous proposition, showing that every conjugate vertex-net
admits an extension to a conjugate face-net.
\begin{proposition}
  \label{prop:conjugatetoconjugateface}
  For every conjugate vertex-net $h: F^{ij}_N \rightarrow \RP^n$ on $F_N^{ij} \simeq V_N$
  there exists a unique conjugate face-net $\lceil h \rceil: F_N \rightarrow \RP^n$,
  such that $\lceil h \rceil$ restricts to $h$ on $F^{ij}_N$. 
\end{proposition}

\begin{proof}
    Without loss of generality, let $i=1, j=2$.
    We claim that the extension of $h$ is given by
    \begin{align}
      \lceil h \rceil (f) \coloneqq \bigcap_{v \sim f} \square h(v). \label{eq:defcompletion}
    \end{align}
    If the four planes defining $\lceil h \rceil(f)$ intersect in a single point,
    then, by construction, $\lceil h \rceil$ is a conjugate face-net that restricts to $h$, and uniqueness is immediate.
    We proceed by showing that Equation~\eqref{eq:defcompletion} defines points for all types of faces $f \in F_N$.
    \begin{itemize}
    \item
      $12$-faces:\\
      For a $12$-face $f$ the intersection $\lceil h \rceil(f)$ is the point $h(f)$.

    \item
      $23$-faces (analogously $kl$-faces with $k=1,2$ and $l \geq 3$):\\
      Consider a $23$-face $\tilde f = (v, v_2, v_{23}, v_3)$ of $\Z^N$ and ``incident'' $13$-face $(f, f_1, f_{13}, f_2)$ of $F_N^{12}$ as in \eqref{eq:13-face} (see Figure~\ref{fig:conjugate-face-net-restriction}, middle).
      The intersection is given by
      \begin{equation}
        \label{eq:extension-focal-point}
        \begin{aligned}
          \lceil h \rceil(\tilde f)
          &= \square h(v) \cap \square h(v_2) \cap \square h(v_{23}) \cap \square h(v_3)\\
          &= (h(f) \vee h(f_1)) \cap  (h(f_3) \vee h(f_{13})).
        \end{aligned}
      \end{equation}
      Since $h$ is a conjugate vertex-net, the two lines are in a common plane, and thus intersect in a unique point (called a \emph{focal point} of $h$).
    \item
      $34$-faces (analogously $kl$-faces with $3 \leq k < l$):\\
      Consider a $34$-face $f = (v, v_3, v_{34}, v_4)$ of $\Z^N$, i.e.,
      \begin{align}
        v,\, v_3,\, v_{34},\, v_4 \ = \ r,\, r+e_3,\, r+e_3+e_4,\, r+e_4
      \end{align}
      with some $r \in V_N$.
      In this case, the intersection $\lceil h \rceil(f)$ is the intersection of the four planes
      \begin{align}
        P_1 &\coloneqq \square h(v),& P_2 &\coloneqq \square h(v_3),& P_3 &\coloneqq \square h(v_{34}), &P_4 &\coloneqq \square h(v_4).
      \end{align}
      The span of the four planes is 4-dimensional, since they span the points of the image of a 4-cube of a conjugate net.
      Note that, analogously as in the proof of Proposition~\ref{prop:conjugatefacetoconjugate},
      in non-generic cases one may invoke Lemma~\ref{lem:dimvertexnets} to work with a lift instead.
      Let us also define the four 3-dimensional spaces
      \begin{align}
        H_1 &\coloneqq P_1 \vee P_2,& H_2 &\coloneqq P_2 \vee P_3,& H_3 &\coloneqq P_3 \vee P_4, &H_4 &\coloneqq P_4 \vee P_1.
      \end{align}
      Clearly, the intersection of these four 3-dimensional spaces coincides with the intersection of the four planes $P_1$, $P_2$, $P_3$, $P_4$.
      Moreover, the intersection of four 3-dimensional spaces in a 4-dimensional space is a single point, which proves the claim.
      \qedhere
    \end{itemize}
\end{proof}

As a consequence of the previous two propositions,
there is a bijection between conjugate vertex-nets and conjugate face-nets,
which together prove Theorem~\ref{th:conjugatebijection}, that is:
\begin{proof}[Proof of Theorem~\ref{th:conjugatebijection}]
  The statement follows directly from Proposition~\ref{prop:conjugatefacetoconjugate} and Proposition~\ref{prop:conjugatetoconjugateface}.
\end{proof}

\begin{remark}
  \label{rem:extension-symmetry}
  Applied to binets, and given two directions $1 \leq i < j \leq N$,
  we may view a conjugate binet $b \cong (g, \lceil h \rceil)$ as the extension of a pair of conjugate vertex-nets $(g, h)$
  defined on $V_N$ and $F_N^{ij}$, respectively.
  Note that by the symmetry of $V_N$ and $F_N^{ij}$, we can interchange their roles.
  This leads to a second way of extending a pair of conjugate vertex-nets.
  We can view $(g, h)$ as a pair of conjugate nets defined on $F_N^{ij}$ and $V_N$, respectively,
  and thus extend it to a conjugate binet $(\lceil g \rceil, h)$ ``in the other way''.

  If we view conjugate binets as \emph{fundamental transformations} of two-dimensional discrete conjugate nets,
  then this observation recovers the symmetry between vertices and faces within the directions of the discrete surface.
  In this interpretation, the symmetry is then broken only in the directions of the transformations.
\end{remark}

\begin{remark}
  \label{rem:focal-points}
  As seen from \eqref{eq:extension-focal-point} in the proof of Proposition~\ref{prop:conjugatetoconjugateface},
  for $N>2$ the extension of a conjugate vertex-net $h$ to a conjugate face-net $\lceil h \rceil$ contains focal points of $h$.
  Conversely, and in particular in comparison with the smooth theory of conjugate systems,
  we may view a conjugate face-net, and hence a conjugate binet,
  as a collection of points of a conjugate system together with some of its focal points
  (see also Figure~\ref{fig:focal-points} and Section~\ref{sec:smooth-vs-discrete}).
\end{remark}

\subsection{Multi-dimensional consistency}
\label{sec:conjugate-binets-consistency}

In this section we show the multi-dimensional consistency of conjugate binets.
Multi-dimensional consistency is a form of discrete integrability \cite{ddgbook},
based on generalizing the definitions of objects on $\Z^2$ to $\Z^N$ with $N>2$.
To do this, one needs to find a suitable definition for a generalization
and then show that it does not lead to contradictions, that is, that the space of such generalized objects is non-empty.
More specifically, a system on $\Z^N$ is a \emph{$k$-dimensional system} if it is uniquely determined by $(k-1)$-dimensional initial data, also known as Cauchy data. The initial data is usually given on the coordinate planes, which then uniquely determines the remaining data on $\Z^k \subset \Z^N$. Moreover, if $(k-1)$-dimensional initial data also uniquely and without contradictions determines solutions on $\Z^m$ for $m > k$, the system is called \emph{multi-dimensionally} consistent.
It is sufficient to show this for $m = k+1$.

The generalization of conjugate nets to $\Z^N$ in terms of quadrilateral lattices, called conjugate vertex-nets in this paper,
is due to Doliwa and Santini \cite{dsqnet}, who also proved their multi-dimensional consistency.
For further reference, let us restate their result in the following proposition.
\begin{proposition}[\cite{dsqnet}] \label{prop:conjugatenetsconsistency}
  Conjugate vertex-nets are multi-dimensionally consistent 3D-systems.
\end{proposition}

We now use the bijection between conjugate face-nets and conjugate vertex-nets,
developed in the previous section,
to deduce the multi-dimensional consistency of conjugate face-nets.
\begin{proposition} \label{prop:conjugatefacenetsconsistency}
  Conjugate face-nets are multi-dimensionally consistent 3D-systems.
\end{proposition}
\begin{proof}
  Since conjugate face-nets are in bijection with conjugate vertex-nets (Theorem~\ref{th:conjugatebijection})
  and conjugate nets are multi-dimensionally consistent (Proposition~\ref{prop:conjugatenetsconsistency}),
  conjugate face-nets are likewise multi-dimensionally consistent.
\end{proof}
\begin{remark}
  \label{rem:conjugate-face-net-propagation}
  The propagation for the 3D-system of conjugate face-nets is best described in terms of its representation by planes (see Definition~\ref{def:planes}~\ref{def:planes-face-net} and Remark~\ref{rem:conjugate-face-net-planes}).
  For an elementary 3-cube, the plane corresponding to the eighth vertex is determined by completing the octahedron formed by the seven other planes.
\end{remark}

Together, this yields the multi-dimensional consistency of conjugate binets.
\begin{proposition}\label{prop:conjugatebinetconsistency}
  Conjugate binets are multi-dimensionally consistent 3D-systems.
\end{proposition}
\begin{proof}
  Since the conjugate vertex-net and the conjugate face-net of a conjugate binet are independent,
  this is an immediate consequence of Proposition~\ref{prop:conjugatenetsconsistency} and Proposition~\ref{prop:conjugatefacenetsconsistency}.
\end{proof}

Finally, combining the three previous propositions, we have proven one of our main results, Theorem~\ref{th:conjugateconsistency}.

\begin{proof}[Proof of Theorem~\ref{th:conjugateconsistency}]
  Combine Proposition~\ref{prop:conjugatenetsconsistency}, Proposition~\ref{prop:conjugatefacenetsconsistency} and Proposition~\ref{prop:conjugatebinetconsistency}.  
\end{proof}

\section{Polar conjugate binets}
\label{sec:polarbinets}

In this section we establish the multi-dimensional consistency of \emph{polar conjugate binets},
which we have introduced in Definition~\ref{def:polar-conjugate}.
Since polar conjugate binets form a special class of conjugate binets,
this can be achieved using another concept of discrete integrability, namely \emph{consistent reductions}.
A reduction is an additional constraint imposed on a $k$-dimensional system.
It is called consistent if the constraint needs to be prescribed only on the initial data,
and is then automatically satisfied on all of $\Z^k$.
If the underlying system is multi-dimensionally consistent, then so is any consistent reduction.

We proceed to prove that polar conjugate binets are a consistent reduction of conjugate binets, that is Theorem~\ref{th:polarconsistency}.

\begin{proof}[Proof of Theorem~\ref{th:polarconsistency}]
  Let $\mathcal{Q}$ be a quadric in $\RP^n$ and let
  $b \colon D_N \to \RP^n$ be a conjugate binet.
  A polar conjugate binet is characterized by the polarity constraint
  \begin{equation}
    b(v) \in b(f)^\pol \quad \text{for all incident } v \in V_N,\ f \in F_N,
  \end{equation}
  or, equivalently, by
  \begin{equation}
    \square b(d) \subset b(d)^\pol \quad \text{for all } d \in D_N,
  \end{equation}
  where $(\cdot)^\pol$ denotes the polar subspace with respect to $\mathcal{Q}$.
  In the following, we use this latter characterization and propagate the binet together
  with its associated planes (see Remark~\ref{rem:conjugate-face-net-propagation}).

  \begin{figure}[t]
    \centering
    \begin{overpic}[width=0.27\textwidth]{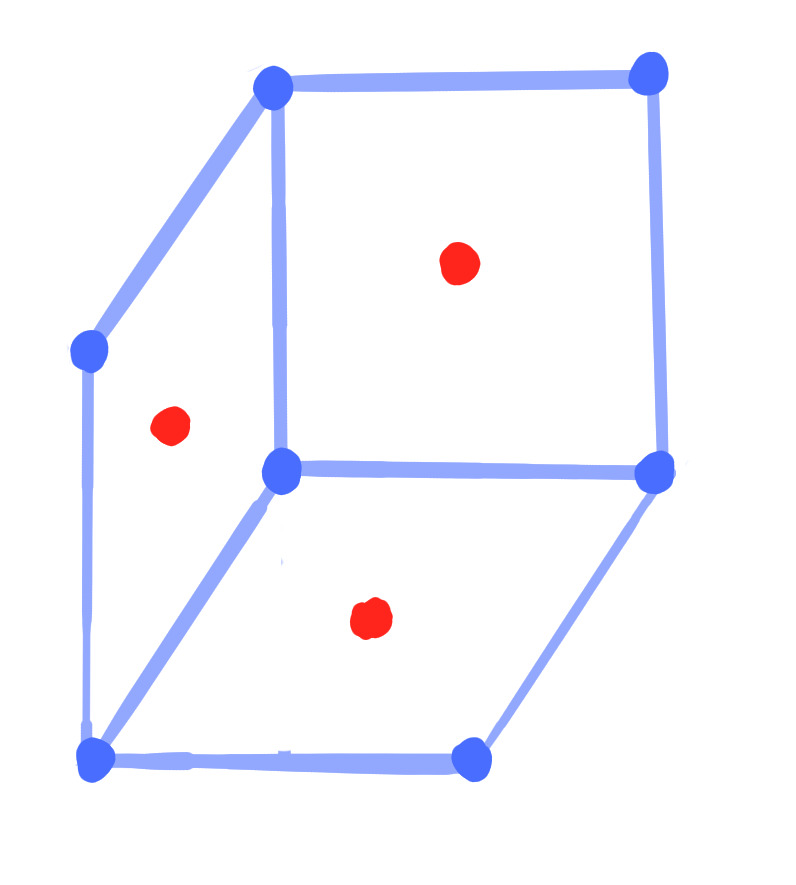}
      \put(34,49){$\color{blue}v$}
      \put(4,6){$\color{blue}v_1$}
      \put(78,45){$\color{blue}v_2$}
      \put(25,95){$\color{blue}v_3$}
      \put(55,6){$\color{blue}v_{12}$}
      \put(76,95){$\color{blue}v_{23}$}
      \put(-3,61){$\color{blue}v_{13}$}
      \put(45,28){$\color{red}f^{12}$}
      \put(55,68){$\color{red}f^{23}$}
      \put(15,42){$\color{red}f^{13}$}
    \end{overpic}
    \hspace{0.1\textwidth}
    \begin{overpic}[width=0.27\textwidth]{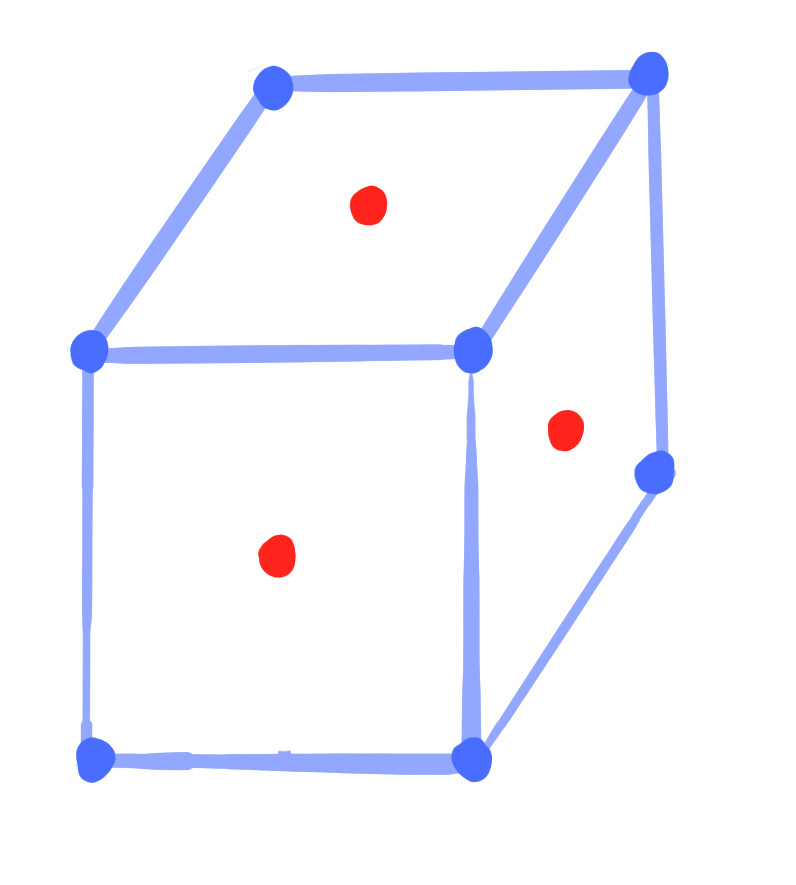}
      \put(57,59){$\color{blue}v_{123}$}
      \put(4,6){$\color{blue}v_1$}
      \put(78,45){$\color{blue}v_2$}
      \put(25,95){$\color{blue}v_3$}
      \put(55,6){$\color{blue}v_{12}$}
      \put(76,95){$\color{blue}v_{23}$}
      \put(-3,61){$\color{blue}v_{13}$}
      \put(44,75){$\color{red}f^{12}_3$}
      \put(34,35){$\color{red}f^{23}_1$}
      \put(59,41){$\color{red}f^{13}_2$}
    \end{overpic}
    \caption{Combinatorics for the consistency around the cube of a polar conjugate binet.}
    \label{fig:consistent-reduction}
  \end{figure}
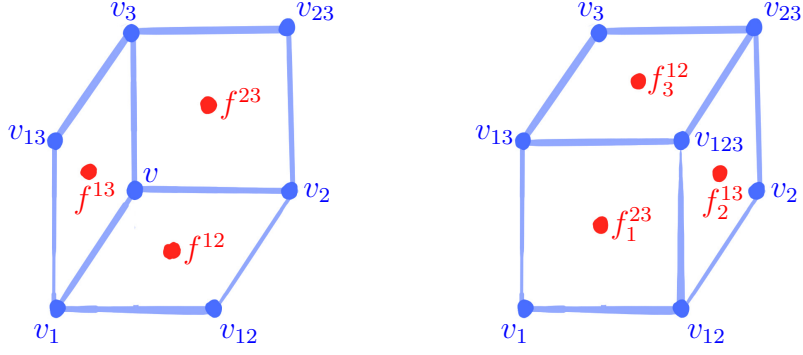

  Consider the eight vertices
  \begin{equation}
    v,\, v_1,\, v_2,\, v_3,\, v_{12},\, v_{23},\, v_{13},\, v_{123}
  \end{equation}
  of a 3-cube in $\Z^N$ and its six faces
  \begin{equation}
    f^{12},\, f^{23},\, f^{13},\, f^{12}_3,\, f^{23}_1,\, f^{13}_2,
  \end{equation}
  as shown in Figure~\ref{fig:consistent-reduction}.
  Assume that the polarity constraint holds at all vertices except $v_{123}$
  and on the three faces $f^{12}, f^{23}, f^{13}$.
  To complete the cube, it suffices to verify that
  \begin{equation}
    \begin{aligned}
      &\square b(f^{12}_3) \subset b(f^{12}_3)^\pol, \quad
      \square b(f^{23}_1) \subset b(f^{23}_1)^\pol, \quad
      \square b(f^{13}_2) \subset b(f^{13}_2)^\pol,\\
      &\square b(v_{123}) \subset b(v_{123})^\pol
    \end{aligned}
    \label{eq:polar-completion}
  \end{equation}
  holds.
  We first consider the face $f^{12}_3$.
  By definition of a conjugate binet,
  \begin{equation}
    \begin{aligned}
      \square b(f^{12}_3) &= b(v_{23}) \vee b(v_3) \vee b(v_{13}),\\
      b(f^{12}_3) &= \square b(v_{23}) \cap \square b(v_3) \cap \square b(v_{13}).
    \end{aligned}
  \end{equation}
  Since the polarity constraint holds on the initial data, we have
  \[
    b(v_{23}) \in \square b(v_{23})^\pol,\quad
    b(v_3) \in \square b(v_3)^\pol,\quad
    b(v_{13}) \in \square b(v_{13})^\pol,
  \]
  and therefore
  \begin{equation}
    \begin{aligned}
      \square b(f^{12}_3)
      &\subset \square b(v_{23})^\pol \vee \square b(v_3)^\pol \vee \square b(v_{13})^\pol\\
      &= \bigl(\square b(v_{23}) \cap \square b(v_3) \cap \square b(v_{13})\bigr)^\pol\\
      &= b(f^{12}_3)^\pol.
    \end{aligned}
  \end{equation}
  The same argument applies to the faces $f^{23}_1$ and $f^{13}_2$.
  Finally, for the vertex $v_{123}$ we have
  \begin{equation}
    \begin{aligned}
      \square b(v_{123}) &= b(f^{12}_3) \vee b(f^{13}_2) \vee b(f^{23}_1),\\
      b(v_{123}) &= \square b(f^{12}_3) \cap \square b(f^{13}_2) \cap \square b(f^{23}_1),
    \end{aligned}
  \end{equation}
  and, using the polarity of the three faces just established,
  \begin{equation}
    \begin{aligned}
      \square b(v_{123})
      &\subset \square b(f^{12}_3)^\pol \vee \square b(f^{13}_2)^\pol \vee \square b(f^{23}_1)^\pol\\
      &= \bigl(\square b(f^{12}_3) \cap \square b(f^{13}_2) \cap \square b(f^{23}_1)\bigr)^\pol\\
      &= b(v_{123})^\pol.
    \end{aligned}
  \end{equation}
  This verifies all conditions in \eqref{eq:polar-completion} and completes the proof.
\end{proof}

\section{Principal binets}
\label{sec:principal-binets}

We have introduced \emph{principal binets} in Definition~\ref{def:principal-binet}.
The orthogonality constraint that relates the two conjugate nets of a principal binet is defined on the crosses $C_N$ of $\Z^N$.
Note that, up to permutation, there is only one cross per edge in $\Z^2$, whereas in $\Z^N$ each edge is contained in $\binom{2N-2}{2}$ crosses (cf.\ Figure~\ref{fig:crosses-types}).

\begin{remark}
  \label{rem:principal-constraints}
  For an edge $(v,v') \in E_N$ of a conjugate binet, the points of all faces $f, f' \in F_N$ such that $(v, f, v', f') \in C_N$ form a cross lie on the line $\square b(v) \cap \square b(v')$.
  Hence, despite the large number of combinatorial crosses, a principal binet effectively imposes a single orthogonality constraint per edge $(v,v') \in E_N$, which involves the two lines $b(v) \vee b(v')$ and $\square b(v) \cap \square b(v')$.  
\end{remark}

In this section we establish that principal binets constitute a consistent reduction of conjugate binets.
Our approach is based on Möbius lifts, as introduced in \cite{atbinets}.

Let $\mobq \subset \RP^{n+1}$ be the \emph{Möbius quadric}, which is the unit-sphere in some affine chart of $\RP^{n+1}$.
Let $B$ be a point on $\mobq$, and let
\begin{equation}
  \pi: \RP^{n+1} \setminus\{B\} \rightarrow \R^n  
\end{equation}
be the standard stereographic projection extended to $\RP^{n+1}$ as a central projection with center $B$.
We recall the following definition of Möbius lifts for principal binets from \cite{atbinets}.
Note that in \cite{atbinets} Möbius lifts were introduced more generally,
but we only require the case of principal binets in the following.
\begin{definition}
  \label{def:moebius-lift}
  Let $b : D_N \rightarrow \R^n$ be a principal binet.
  A binet $b_\mobq : D_N \rightarrow \RP^{n+1}$ is called a \emph{Möbius lift} of $b$ if
  \begin{enumerate}
  \item
    $b_\mobq$ is a polar conjugate binet with respect to $\mobq$,
  \item
    $b$ is the projection of $b_\mobq$, that is $\pi \circ b_\mobq = b$.
    \qedhere
  \end{enumerate}
\end{definition}

Let us recollect two relevant properties of Möbius lifts.

\begin{lemma}\
  \label{lem:moebius-lifts}
  \nobreakpar
  \begin{enumerate}
  \item
    \label{lem:moebius-lifts-existence}
    Every principal binet has a one-parameter family of Möbius lifts.
  \item
    \label{lem:moebius-lifts-projection}
    The projection $\pi$ of any polar conjugate binet (with respect to $\mobq$) is a principal binet.
    \qedhere
  \end{enumerate}
\end{lemma}
\begin{remark}
  \label{rem:moebius-lifts-combinatorics}
  The statements of Lemma~\ref{lem:moebius-lifts} were shown in \cite{atbinets} for the case $N=2$.
  The proof only requires considering a single cross (a single quad of the double graph $D_N$),
  so the results extend immediately to arbitrary dimension $N$ and also to simply connected quad-graphs.
\end{remark}
\begin{remark}
  \label{rem:orthogonal-sphere-representation}
  The points of the Möbius lift may be interpreted as spheres centered at the points of the principal binet.
  They constitute an \emph{orthogonal sphere representation} for the principal binet \cite{atbinets}.
  In the one-parameter family of Möbius lifts, the radius of one sphere may be chosen arbitrarily
  and then the radii of all other spheres are determined by the orthogonality constraints.
\end{remark}

Using these notions, we are in a position to prove Theorem~\ref{th:principalconsistency},
the main result on the multi-dimensional consistency of principal binets.

\begin{proof}[Proof of Theorem~\ref{th:principalconsistency}]
  Let $b : D_3 \rightarrow \R^n$ be a conjugate binet.
  Assume that on the union of the 2-dimensional coordinate planes, $b$ satisfies the constraints defining a principal binet.
  Note that this implies additional conditions at the corner where the coordinate planes meet. 
  By Lemma~\ref{lem:moebius-lifts}~\ref{lem:moebius-lifts-existence} and Remark~\ref{rem:moebius-lifts-combinatorics}, we may choose a Möbius lift $b_\mobq$ on the initial data.
  The Möbius lift is a polar conjugate binet in $\RP^{n+1}$.
  Therefore, by Theorem \ref{th:polarconsistency}, we can extend the Möbius lift from the initial data to a polar conjugate binet $b_\mobq$ on all of $D_3$.
  Since both $b$ and $b_\mobq$ are conjugate binets and are uniquely determined by their initial data, we have $\pi \circ b_\mobq = b$ on $D_3$.
  Thus, by Lemma~\ref{lem:moebius-lifts}~\ref{lem:moebius-lifts-projection}, $b$ is a principal binet on all of $D_3$.
\end{proof}

\begin{remark}
  Since Theorem~\ref{th:polarconsistency} holds for arbitrary quadrics,
  the result on the consistency of princpal binets in Theorem~\ref{th:principalconsistency} may be transferred to other geometries
  (including hyperbolic, elliptic, Lorentz, isotropic, de Sitter, and anti-de Sitter)
  by an appropriate choice of quadric and lift.
\end{remark}

Given two directions $1 \leq i < j \leq N$, let
$g : V_N \rightarrow \R^n$ and $h : F_N^{ij} \cong V_N \rightarrow \R^n$
be two conjugate vertex-nets.
We denote by $(g, \lceil h \rceil)$ the conjugate binet obtained from $g$
together with the extension of $h$ to a conjugate face-net
$\lceil h \rceil : F_N \rightarrow \R^n$
(cf.\ Theorem~\ref{th:conjugatebijection}).
As explained in Remark~\ref{rem:extension-symmetry},
interchanging the roles of $V_N$ and $F_N^{ij}$ instead leads to an extension
$\lceil g \rceil : F_N \rightarrow \R^n$ of $g$,
and hence to the conjugate binet $(\lceil g \rceil, h)$.
Principal binets exhibit the following symmetry with respect to this interchange.
\begin{theorem}
  \label{th:principalsymmetry}
  Let $1 \leq i < j \leq N$,
  and let $g, h$ be two conjugate vertex-nets defined on $V_N$ and $F_N^{ij}$, respectively.
  Then the conjugate binet $(g,\lceil h\rceil)$ is a principal binet
  if and only if the conjugate binet $(\lceil g\rceil, h)$ is a principal binet.
\end{theorem}
\begin{proof}
  Let $b \coloneqq (g,\lceil h\rceil)$ and $b' \coloneqq (\lceil g\rceil, h)$ be the two conjugate binets obtained as extensions of $(g, h)$.
  Assume $b$ is a principal binet.
  Then, by Lemma~\ref{lem:moebius-lifts}~\ref{lem:moebius-lifts-existence}, it has a Möbius lift $b_\mobq$, which is a polar conjugate binet.
  The restriction of $b_\mobq$ to $V_N \cup F^{ij}_N$ yields a pair of conjugate vertex-nets $(g_\mobq, h_\mobq)$.
  Since all planes $\square \lceil h_\mobq \rceil = \square h_\mobq$ and thus their points of intersection are already uniquely determined by $h_\mobq$, it is straightforward to verify that $b_\mobq = (g_\mobq, \lceil h_\mobq \rceil)$.
  Moreover, since $b_\mobq$ is a polar conjugate binet, we have
  \begin{equation}
    \begin{aligned}
      \square h_\mobq(v) &\subset g_\mobq(v)^\pol &&\text{for all~}v \in V_N, \text{~and}\\
      \square g_\mobq(f) &\subset h_\mobq(f)^\pol &&\text{for all~}f \in F_N.
    \end{aligned}
  \end{equation}
  
  Next, we define $b'_\mobq$ as the other extension of the Möbius lift $(g_\mobq, h_\mobq)$, that is
  \begin{align}
    b'_\mobq \coloneqq (\lceil g_\mobq \rceil, h_\mobq).
  \end{align}
  It is clear that $\pi \circ b'_\mobq = b'$.
  We show that $b'_\mobq$ is a Möbius lift of $b'$.
  By Proposition~\ref{prop:conjugatetoconjugateface}, $b'_\mobq$ is a conjugate binet.
  Therefore it remains to show that $b'_\mobq$ is a polar binet.

  \begin{figure}[t]
    \centering
    \includegraphics[width=0.28\textwidth]{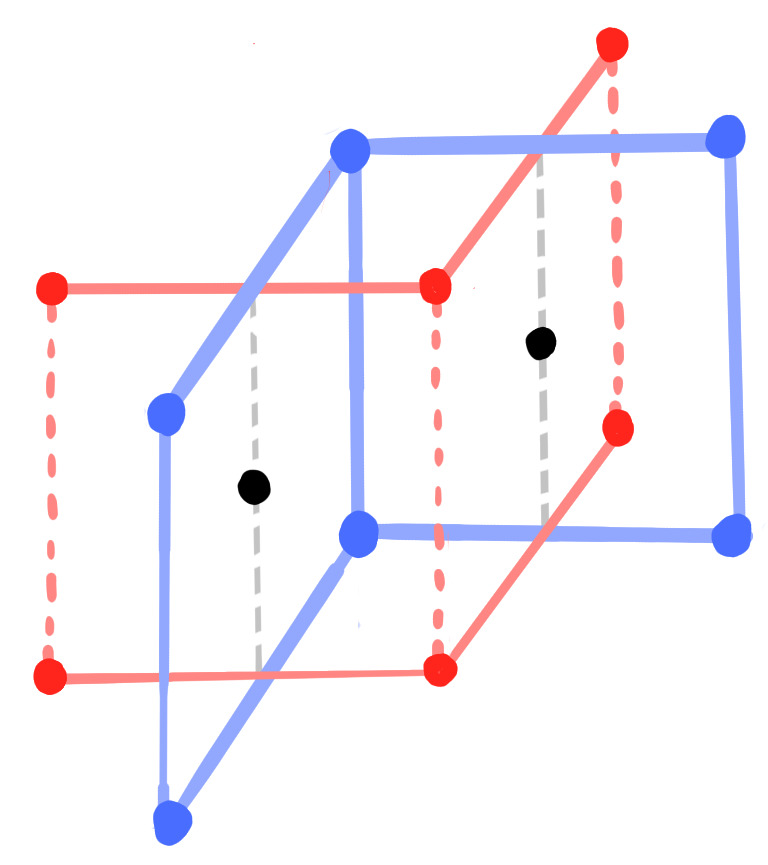}
    \caption{
      Vertices $V_3$ (blue) and $12$-faces $F_3^{12}$ (red) of $\Z^3$.
      Upon identification of $F_3^{12}$ with the vertices of another $\Z^3$,
      the faces of this $\Z^3$ are naturally identified with $W_3 = V_3 \cup (F_3\setminus F_3^{12})$ (blue and black).
    }
    \label{fig:symmetric-extension}
  \end{figure}
  If we identify $F_N^{ij}$ with the vertices of $\Z^N$, then the faces of this $\Z^N$ may be naturally identified with
  \begin{equation}
    W_N \coloneqq V_N \cup (F_N \setminus F_N^{ij})
  \end{equation}
  (see Figure~\ref{fig:symmetric-extension}).
  Then the extension of $g_\mobq$ is given by
  \begin{align}
    \lceil g_\mobq\rceil(w)   = \bigcap_{f \in F_N^{ij},\ f\sim w} \square g_\mobq(f)\qquad\text{for all}~w \in W_N.
  \end{align}
  Thus, for any incident $w \in W_N$ and $f \in F_N^{ij}$, we have
  \begin{align}
    \lceil g_\mobq\rceil(w) = \bigcap_{\tilde f \in F_N^{ij},\ \tilde f\sim w} \square g_\mobq(\tilde f) \subset \square g_\mobq(f) \subset h_\mobq(f)^\pol.
  \end{align}
  Hence, $b'_\mobq$ is a Möbius lift of $b'$,
  and by Lemma~\ref{lem:moebius-lifts}~\ref{lem:moebius-lifts-projection}, it follows that $b'$ is a principal binet.
\end{proof}

Thus, a pair of conjugate vertex-nets
\begin{equation}
  g : V_N \rightarrow \R^n, \qquad
  h : F_N^{ij} \rightarrow \R^n
\end{equation}
define a principal binet, if they possess a Möbius lift $(g_\mobq, h_\mobq)$ as in the proof,
that is a pair of conjugate vertex-nets
\begin{equation}
  g_\mobq : V_N \rightarrow \RP^{n+1},\qquad
  h_\mobq : F_N^{ij} \rightarrow \RP^{n+1}  
\end{equation}
such that
\begin{enumerate}
\item $\pi \circ g_\mobq = g$, $\pi \circ h_\mobq = h$, and
\item $g(v)$ and $h(f)$ are related by polarity (with respect to $\mobq$) for all incident $v \in V_N$ and $f \in F_N^{ij}$.
\end{enumerate}
In this way, the relation between $g$ and $h$ becomes completely symmetric.
However, the pair $(g,h)$ is not defined symmetrically on $\Z^N$, since the $ij$-directions are distinguished.
On the other hand, when binets on $\Z^N$ are interpreted as \emph{Ribaucour transformations} of discrete surfaces,
it is natural for certain directions to play a distinguished role (see also Remark~\ref{rem:extension-symmetry} and Section~\ref{sec:circular-nets}).

\begin{remark}
  \label{rem:more-focal-points}
  By Remark~\ref{rem:focal-points}, the points of the extension $(g,\lceil h\rceil)$ may be interpreted as focal points of the pair $(g,h)$.
  Passing instead to the extension $(\lceil g\rceil,h)$ amounts to introducing additional focal points to the conjugate system.
  In comparison with the smooth theory, the resulting symmetry in imposing orthogonality conditions with respect to different focal points
  is investigated further in Section~\ref{sec:smooth-vs-discrete}.
\end{remark}

\begin{question}
  Does there exist a simple, direct geometric condition on the two conjugate vertex-nets $(g, h)$
  that is equivalent to the pair $(g,\lceil h\rceil)$ (or equivalently $(\lceil g\rceil, h)$) forming a principal binet?
\end{question}

\begin{question}
  For each pair of directions, we obtain two possible extensions to a principal binet.
  Do these extensions form a ``compatible'' system of principal binets in a meaningful sense?
  Furthermore, among the resulting points, which should be considered points of the conjugate system and which should be regarded as focal points?
\end{question}

\section{Circular, conical and principal contact element nets}
\label{sec:circular-conical}

In this section, we briefly revisit earlier discretizations of principal nets and their relation to principal binets.

\subsection{Circular nets and Ribaucour  transformations}
\label{sec:circular-nets}

Let $b : D_N \rightarrow \R^n$ be a principal binet with M\"obius lift $b_\mobq : D_N \rightarrow \RP^{n+1}$.
The points of $b_\mobq$ can be interpreted as spheres centered at the points of $b$, an interpretation known as the orthogonal sphere representation (see Remark~\ref{rem:orthogonal-sphere-representation}).
Furthermore, for each quad $f \in F_N$, the intersection of the plane $\square b_\mobq(f)$ with the M\"obius quadric $\mobq$ defines a circle $C(f)$.
By construction, this circle is orthogonal to the spheres corresponding to the vertices incident with $f$.
These circles generalize the circles of a circular net.
Indeed, if the M\"obius lift lies entirely on $\mobq$, the spheres of the orthogonal sphere representation degenerate to points, and the four points of each quad are contained in the corresponding circle.
This relation to circular nets was discussed in \cite{atbinets} for the case $N=2$ and immediately generalizes to arbitrary dimensions.

For $N>2$, we can recover further geometric structures analogous to those of circular nets.
For the M\"obius lift of a principal binet, the span of each elementary 3-cube is a three-dimensional projective subspace.
The intersection of this span with the M\"obius quadric therefore defines a sphere.
In the special case of circular nets, this sphere contains the eight points corresponding to the vertices of the cube.
In the general case of a principal binet, it is a sphere orthogonal to the eight spheres of the orthogonal sphere representation corresponding to the vertices of the cube.
Furthermore, the axes of the six circles corresponding to the faces of the cube intersect at the center of that sphere (see Figure~\ref{fig:central-sphere}).

The spheres associated with each 3-cube of a principal binet also occur naturally in the comparison with discrete orthogonal coordinate systems (see Section~\ref{sec:toc}).
On the other hand, if we view one of the directions in a principal binet as the direction of a fundamental transformation, then the circles involving that direction discretize the touching circles of a \emph{Ribaucour transformation}.
Similarly, the spheres associated with each 3-cube discretize the touching spheres of this Ribaucour transformation.

\begin{figure}[t]
    \centering
    \includegraphics[width=0.49\textwidth]{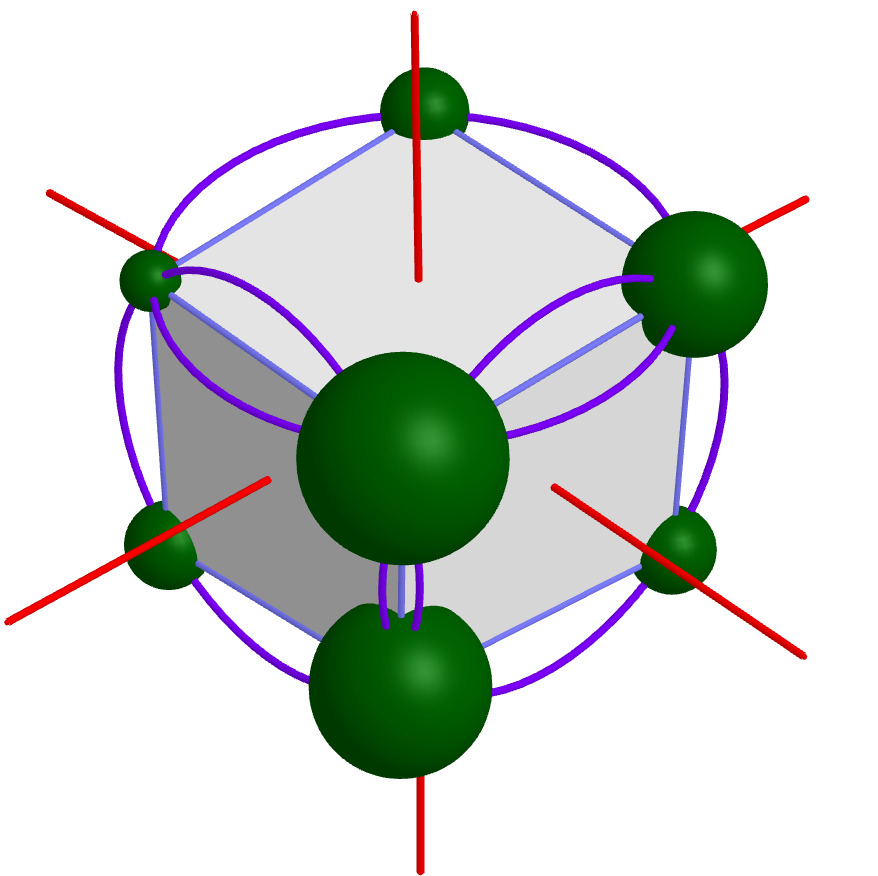}
    \includegraphics[width=0.49\textwidth]{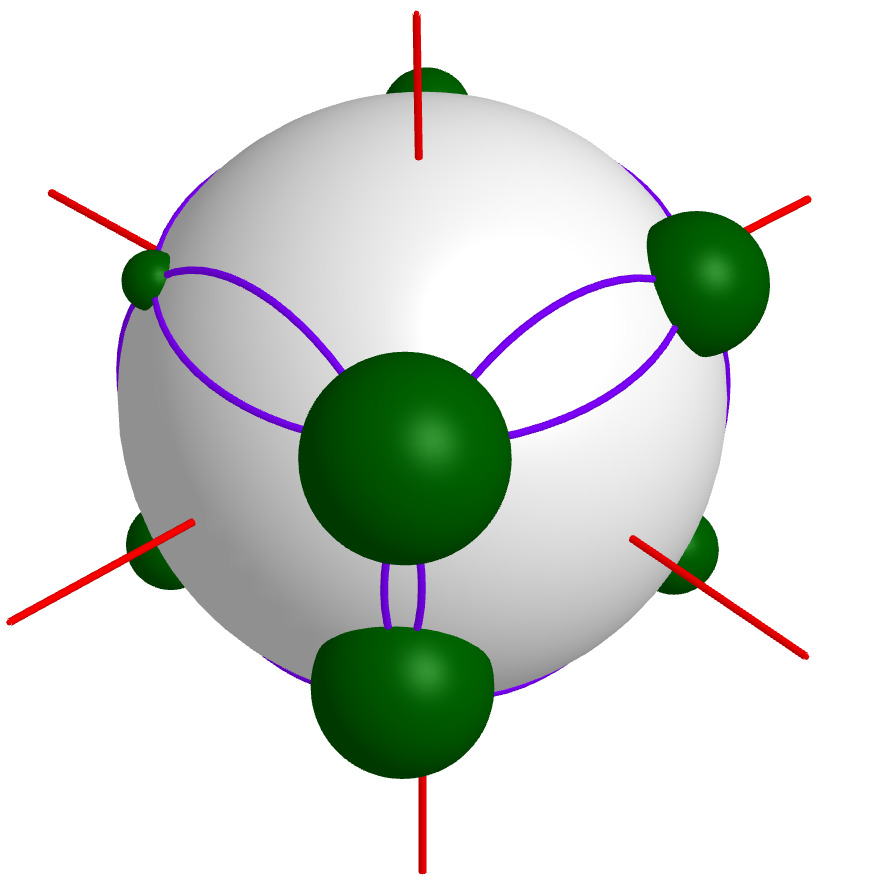}
    \caption{An elementary cube of a principal binet. The restriction to $\Z^3$ is drawn in blue (left only), the circles in violet with axes in red, the spheres of the orthogonal sphere representation on $\Z^3$ in green, and the central sphere in white (right only).}
    \label{fig:central-sphere}
\end{figure}

\subsection{Conical nets as a consistent reduction}
\label{sec:conical-nets}

Similar to circular nets, conical nets can also be viewed as a special case of principal binets with $n=3$ using a Laguerre geometric description \cite{atbinets}.
Conical nets are naturally described as maps assigning to every vertex of $\Z^N$ an oriented plane in $\R^3$.
They are characterized by the condition that the four planes of any elementary quad are in oriented contact with a cone of revolution \cite{lpwywconical, ddgbook}.
In this formulation, conical nets correspond to conjugate vertex-nets inscribed in the Blaschke cylinder \cite{bsorganizing,ddgbook}.
Hence, the multi-dimensional consistency of conjugate vertex-nets in quadrics \cite{doliwaqnetsquadrics} implies the multi-dimensional consistency of conical nets.
For $N>2$, conical nets cannot consistently be viewed as maps of points defined on the vertices of $\Z^N$.
However, if we instead define conical nets via the intersection points of the planes associated with the faces of $\Z^N$, the resulting combinatorics for $N>2$ coincides with that in our definition of conjugate face-nets (Definition~\ref{def:conjugate}~\ref{itm:conjugateface}).
In this way, we can view them as consistent reductions of conjugate face-nets:
\begin{theorem}
  \label{thm:conical-nets-as-face-nets}
  Conical nets are a consistent reduction of conjugate face-nets.
\end{theorem}
\begin{proof}
  Follows from the multi-dimensional consistency of conjugate face-nets (Proposition~\ref{prop:conjugatefacenetsconsistency})
  and the multi-dimensional consistency of conical nets \cite{bsorganizing,ddgbook,doliwaqnetsquadrics}.
\end{proof}

\subsection{Principal contact element nets and line complexes}
\label{sec:principal-contact-element-nets}

Principal contact element nets can be understood as \emph{isotropic line complexes} in the Lie quadric (note that on $\Z^2$ they are called \emph{isotropic line congruences}) \cite{bsorganizing}. The consistency of isotropic line complexes is discussed in \cite[Section~3.3]{ddgbook}, though a previous special case corresponding to \emph{A-nets} was investigated earlier \cite{doliwaanetspluecker}.
Line congruences in relation to conjugate vertex-nets were introduced in \cite{dsmlinecongruence}, and they are known to be multi-dimensionally consistent \cite{bslinecomplexes}.

A key insight of \cite{atbinets} is that principal binets admit a Lie geometric description:
they lift to \emph{polar line bicongruences} with respect to the Lie quadric
$\lieq \subset \RP^5$.
More precisely, a principal binet $b : D_2 \rightarrow \R^3$ gives rise to a map
\begin{equation}
  b_\lieq : D_2 \rightarrow \mathrm{Lines}(\RP^5),
\end{equation}
such that
\begin{enumerate}
\item the restriction to $V_2$ is a line congruence (lines of adjacent vertices intersect),
\item the restriction to $F_2$ is a line congruence (lines of adjacent faces intersect),
\item for all incident $d,d' \in D_2$ the lines $b_\lieq(d)$ and $b_\lieq(d')$ are related by polarity with respect to $\lieq$.
\end{enumerate}
Viewed in this way, the polar line bicongruences associated with principal binets
naturally generalize the isotropic line congruences arising from principal contact element nets.
The polarity condition implies that the values of $b_\lieq$ on faces are uniquely determined by its values on vertices:
\begin{equation}
  b_\lieq(f) = \bigcap_{v \sim f} \bigl(b_\lieq(v)\bigr)^\perp .
  \label{eq:polarlinecomplex}
\end{equation}

In the present paper, we have shown that principal binets form a consistent reduction of conjugate binets (Theorem~\ref{th:principalconsistency}) by using the multi-dimensional consistency of polar conjugate binets (Theorem~\ref{th:polarconsistency}).
We note, however, that there is also an alternative approach based on the multi-dimensional consistency of line complexes.
This alternative proof was not adopted for several reasons: it requires substantially more preliminaries, it is less amenable to the other observations made in this paper, and it is, per se, limited to the $\R^3$ case.
We briefly sketch the idea of this alternative proof in the remainder of this section.

The Lie lift of principal binets can be extended to higher dimensions as follows.
A principal binet $b \colon D_N \rightarrow \mathbb{R}^3$ gives rise to a map
\[
b_{\lieq} \colon D_N \rightarrow \mathrm{Lines}(\mathbb{RP}^5),
\]
such that
\begin{enumerate}
\item the restriction to $V_N$ is a line complex,
\item the restriction to $F_N$ has the property that the lines of adjacent faces intersect,
\item for all incident $d,d' \in D_N$ the lines $b_\lieq(d)$ and $b_\lieq(d')$ are related by polarity with respect to $\lieq$.
\end{enumerate}

The relation \eqref{eq:polarlinecomplex} remains valid for $N>2$.
Consequently, the multi-dimensional consistency of line complexes implies the multi-dimensional consistency of such Lie lifts, which in turn yields the consistency of principal binets.

\section{Discrete orthogonal coordinate systems}
\label{sec:toc}

A key motivation for introducing principal binets in \cite{atbinets} comes from \emph{discrete confocal quadrics} \cite{bsstconfocali, bsstconfocalii} (see Figure~\ref{fig:discrete-confocal}, left),
which can be regarded as a special case of discrete orthogonal coordinate systems.
For $N>2$, however, the underlying combinatorial structure differs from that used for principal binets.
In this section, we briefly recall the combinatorial setup and orthogonality constraints for discrete confocal quadrics in the case $N=3$ and explain how they relate to principal binets.

\begin{figure}[t]
  \centering
  \includegraphics[width=0.5\textwidth]{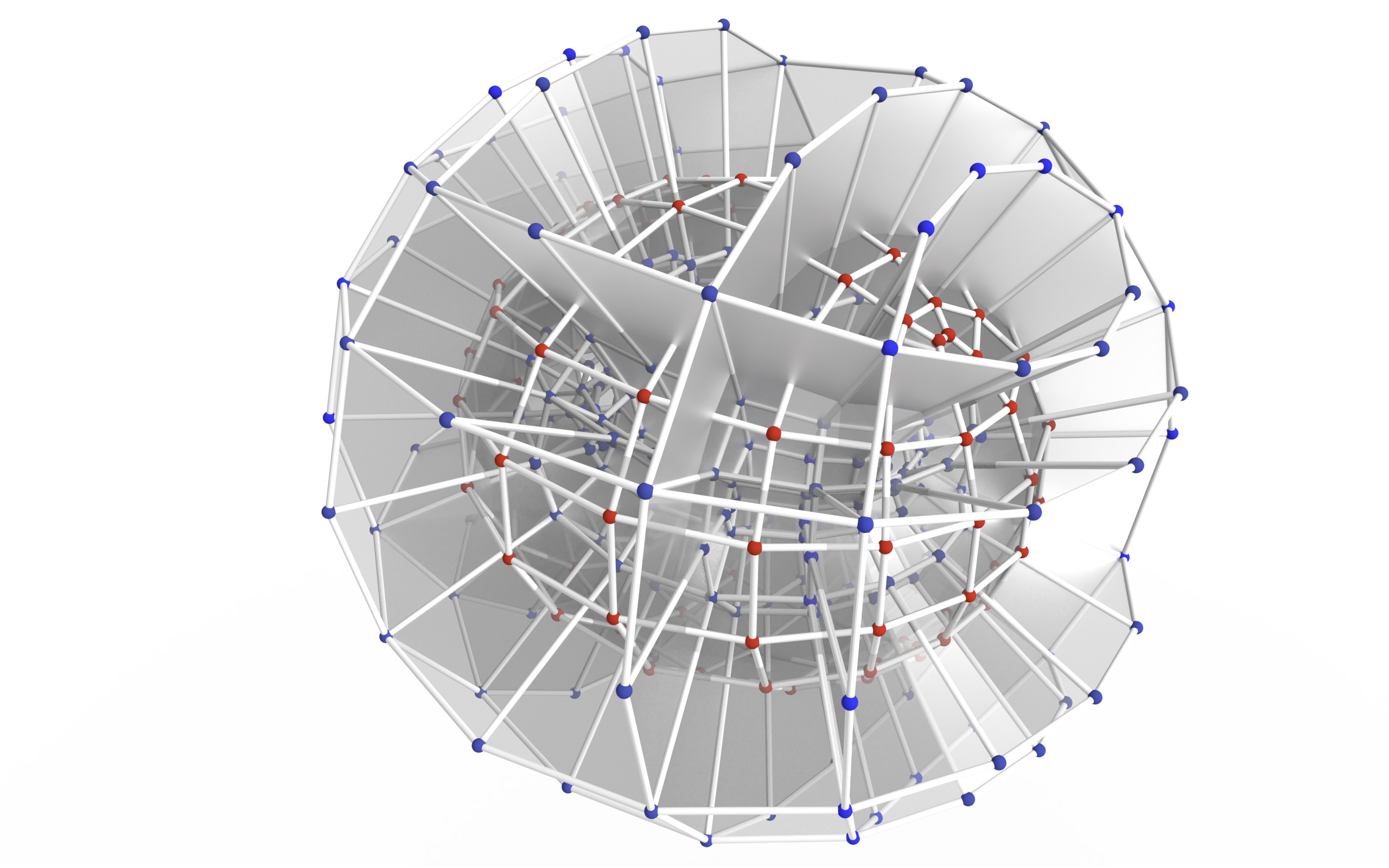}
  \hspace{1cm}
  \includegraphics[width=0.3\textwidth]{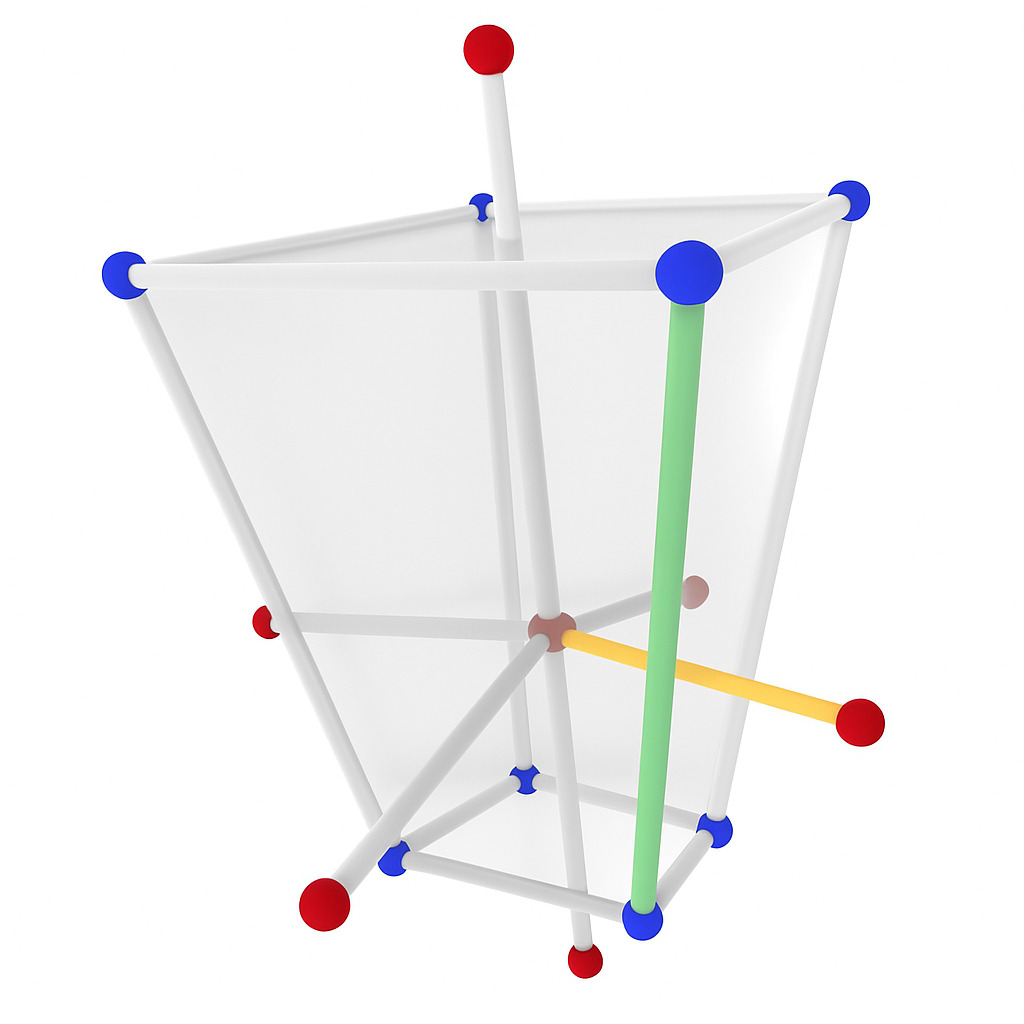}
  \caption{
    \emph{Left}: Discrete confocal quadrics on $\Z^3$ ($V_3$ in blue, $\cube_3$ in red).
    \emph{Right}: Elementary cube of a discrete orthogonal coordinate system.
    The green and yellow edges are an example of an orthogonal pair of edges.
  }
  \label{fig:discrete-confocal}
\end{figure}

In \cite{bsstconfocalii}, discrete confocal quadrics are defined as maps $(\frac{1}{2}\Z)^N \rightarrow \R^N$.
However, the orthogonality constraints involve only the restrictions of these maps to pairs of dual sublattices, such as $\Z^N$ and $(\Z^N)^* = (\Z + \frac{1}{2})^N$.
In the case $N=3$, we can naturally identify $(\Z^3)^*$ with the set of 3-cubes in $\Z^3$, which we denote by $\cube_3$.
From this perspective, the following definition of discrete orthogonal coordinate systems can be extracted from~\cite{bsstconfocalii} (see also \cite{Techter2021}):

\begin{definition}
  \label{def:docs}
A \emph{(three-dimensional) discrete orthogonal coordinate system (dOCS)} is a map $x: V_3 \cup \cube_3 \rightarrow \R^3$ such that, for all pairs of adjacent vertices $v, v' \in V_3$ and all pairs of adjacent cubes $c, c' \in \cube_3$ that share a common face (i.e., $v, v' \in c, c'$), the two lines $x(v) \vee x(v')$ and $x(c) \vee x(c')$ are orthogonal (see Figure~\ref{fig:discrete-confocal}, right).
\end{definition}

An immediate consequence of Definition~\ref{def:docs} is that the image of every face (and dual face) lies in a plane.
Hence, the restriction of a dOCS to $V_3$ (or to $\cube_3 \cong V_3$, respectively) is a conjugate vertex-net.

The notion of a Möbius lift for principal binets (Definition~\ref{def:moebius-lift}) can be adapted to dOCS as follows \cite{Techter2021}:

\begin{definition}
  \label{def:docs-moebius-lift}
  Let $x : V_3 \cup \cube_3 \rightarrow \R^3$ be a dOCS.
  A map $x_\mobq : V_3 \cup \cube_3 \rightarrow \RP^{4}$ is called a \emph{Möbius lift} of $x$ if
  \begin{enumerate}
  \item
    \label{def:docs-moebius-lift-polarity}
    for every incident $v \in V_3$ and $c \in \cube_3$, the points $x(v)$ and $x(c)$ are related by polarity with respect to $\mobq$,
  \item
    $x$ is the projection of $x_\mobq$, that is, $\pi \circ x_\mobq = x$.
    \qedhere
  \end{enumerate}
\end{definition}

The polarity constraint immediately implies that the restriction of a Möbius lift $x_\mobq$ of a dOCS to $V_3$ (and to $\cube_3 \cong V_3$, respectively) is again a conjugate vertex-net.
The existence and completeness results for Möbius lifts of principal binets (Lemma~\ref{lem:moebius-lifts}) extend directly to Möbius lifts of dOCS \cite{Techter2021}.
\begin{lemma}\
  \label{lem:docs-moebius-lift}
  \nobreakpar
  \begin{enumerate}
  \item
    Every dOCS has a one-parameter family of Möbius lifts.
  \item
    \label{lem:docs-moebius-lift-projection}
    The projection $\pi$ of any map $V_3 \cup \cube_3 \rightarrow \RP^4$ that satisfies the polarity constraint~\ref{def:docs-moebius-lift-polarity} from Definition~\ref{def:docs-moebius-lift} is a dOCS.
    \qedhere
  \end{enumerate}
\end{lemma}

Moreover, the polarity constraint \ref{def:docs-moebius-lift-polarity} from Definition~\ref{def:docs-moebius-lift}
implies that the restriction of $x_\mobq$ to $V_3$ (or, respectively, to $\cube_3$)
uniquely determines the entire lift via
\begin{equation}
  x_\mobq(c) = \left(\bigvee_{v \sim c} x(v)\right)^\perp .
\end{equation}
Hence, on the level of Möbius lifts, there is an immediate correspondence between principal binets and dOCS:
starting from a conjugate vertex-net on $V_3$, one may extend either to faces $F_3$ or to cubes $\cube_3$.
The extension to $\cube_3$ is unique, whereas the extension to $F_3$ depends on initial data.
We now make this correspondence explicit and describe its effect on the projections.

Let $b_\mobq \colon V_3 \cup F_3 \rightarrow \RP^4$ be a Möbius lift of a principal binet
$b \colon V_3 \cup F_3 \rightarrow \R^3$.
Since the restriction of $b_\mobq$ to $V_3$ is a conjugate vertex-net, the span of the vertices of every elementary 3-cube is three-dimensional.
Consequently, its polar subspace is a single point.
Geometrically, this point can be interpreted as the center of a sphere associated with each 3-cube (see Section~\ref{sec:circular-nets}).
We therefore define a map $x_\mobq \colon V_3 \cup \cube_3 \rightarrow \RP^4$ by
\begin{equation}
  \begin{aligned}
    x_\mobq(v) &\coloneqq b_\mobq(v)
      && \text{for all } v \in V_3,\\
    x_\mobq(c) &\coloneqq \left(\bigvee_{v \sim c} b_\mobq(v)\right)^\perp
      && \text{for all } c \in \cube_3.
  \end{aligned}
\end{equation}
By construction, $x_\mobq$ satisfies the polarity constraint \ref{def:docs-moebius-lift-polarity} from Definition~\ref{def:docs-moebius-lift} and therefore defines,
by Lemma~\ref{lem:docs-moebius-lift}~\ref{lem:docs-moebius-lift-projection},
a Möbius lift of a dOCS.

We now describe how this dOCS arises directly from the principal binet $b$.
The point $x_\mobq(c)$ is the intersection point of all lines
$x_\mobq(c) \vee x_\mobq(c')$ with $c' \in \cube_3$ adjacent to $c$.
Each such line is polar to the plane $\square b_\mobq(f)$,
where $f \in F_3$ is the face incident to $c$ and $c'$.
Under projection, this line becomes the normal line of the principal binet (or, equivalently, the axis of the circle $C(f)$ introduced in Section~\ref{sec:circular-nets}),
\begin{equation}
  \pi\bigl(x_\mobq(c) \vee x_\mobq(c')\bigr) = N(f),
\end{equation}
where $N(f)$ denotes the line through $b(f)$ orthogonal to $\square b(f)$.
Consequently, the dOCS $x \colon V_3 \cup \cube_3 \rightarrow \R^3$ associated with $b$ is given by
\begin{equation}
  \label{eq:principal-binet-to-docs}
  \begin{aligned}
    x(v) &\coloneqq b(v)
      && \text{for all } v \in V_3,\\
    x(c) &\coloneqq \bigcap_{f \sim c} N(f)
      && \text{for all } c \in \cube_3.
  \end{aligned}
\end{equation}
In particular, the resulting points of $x$ depend only on the principal binet $b$
and not on the choice of Möbius lift.

Conversely, let $x_\mobq : V_3 \cup \cube_3 \rightarrow \RP^4$ be the Möbius lift of a dOCS
$x : V_3 \cup \cube_3 \rightarrow \R^3$.
First, we set
\begin{equation}
  b_\mobq(v) \coloneqq x_\mobq(v)
  \qquad \text{for all } v \in V_3.
\end{equation}
Furthermore, by the polarity condition on the Möbius lift, we have
\begin{equation}
  b_\mobq(f) \in \bigl(\square b_\mobq(f)\bigr)^\perp
  = x_\mobq(c) \vee x_\mobq(c')
  \qquad \text{for all } f \in F_3,
\end{equation}
where $c, c' \in \cube_3$ are the two cubes incident to $f$.
These points may be chosen arbitrarily on one-dimensional initial data,
for instance along coordinate axes in the coordinate planes.
All remaining points in the coordinate planes are then uniquely determined on these lines
by the conjugacy condition, as described in \cite[Section~20]{atbinets}.
Finally, by Theorem~\ref{th:polarconsistency}, polar conjugate binets form a consistent reduction of conjugate binets,
and thus all remaining points in $F_3$ are uniquely determined by the 3D-system of conjugate face-nets.

We now describe how this principal binet arises directly from the dOCS $x$.
The construction carries over directly to the projection.
First, we set
\begin{equation}
  \label{eq:docs-to-principal-binets1}
  b(v) \coloneqq x(v)
  \qquad \text{for all } v \in V_3.
\end{equation}
We then construct the remaining points
\begin{equation}
  \label{eq:docs-to-principal-binets2}
  b(f) \in x(c) \vee x(c')
  \qquad \text{for all } f \in F_3,
\end{equation}
from one-dimensional initial data,
where $c, c' \in \cube_3$ are the two cubes incident to $f$.
The construction proceeds first to two-dimensional data by the conjugacy condition,
as described in \cite[Section~20]{atbinets},
and then to all of $F_3$ using the 3D-system of conjugate face-nets.
This yields a principal binet by Theorem~\ref{th:principalconsistency}.

Note that although the constructions for the Möbius lift and for the projection look identical,
for the Möbius lift the restriction to $V_3$ uniquely determines all remaining points,
while for the projection the restriction to $F_3$ is still used.

We summarize this correspondence in the following theorem.

\begin{theorem}
  \label{thm:binets-docs}
  There is a natural correspondence between principal binets $b : V_3 \cup F_3 \rightarrow \R^3$
  and dOCS $x : V_3 \cup \cube_3 \rightarrow \R^3$ on $\Z^3$:
  \begin{enumerate}
  \item
    Every principal binet
    defines a unique dOCS
    by \eqref{eq:principal-binet-to-docs}.
  \item
    Every dOCS
    defines a principal binet
    by \eqref{eq:docs-to-principal-binets1} and
    by \eqref{eq:docs-to-principal-binets2} on one-dimensional initial data, and propagation as described above (via the
    3D-system of conjugate face-nets).\qedhere
  \end{enumerate}
\end{theorem}

The definition of three-dimensional dOCS (Definition~\ref{def:docs})
can be generalized literally to maps
\begin{equation}
  x : V_N \cup \cube_N \rightarrow \R^3
\end{equation}
in arbitrary dimensions $N \geq 3$, where $\cube_N$ denotes the set of $3$-cubes in $\Z^N$.
Note that this is not the same as the notion of higher-dimensional dOCS used in
\cite{bsstconfocali, bsstconfocalii},
but rather remains fundamentally based on three dimensions.
However, with this definition, together with the previous observation that on the level of the Möbius lift
a dOCS is uniquely determined by its restriction to $V_N$,
and using the multi-dimensional consistency of conjugate vertex-nets
(Proposition~\ref{prop:conjugatenetsconsistency}, \cite{dsqnet}),
we obtain the following new observation.
\begin{theorem}
  \label{thm:docs-consistency}
  dOCS are multi-dimensionally consistent 3D-systems.
\end{theorem}

\begin{remark}
  Theorem~\ref{thm:docs-consistency}, together with the correspondence
  established in Theorem~\ref{thm:binets-docs}, provides an alternative approach
  to proving the multi-dimensional consistency of principal binets
  (Theorem~\ref{th:principalconsistency}) in the case of $\R^3$.
\end{remark}

\section{Smooth orthogonal coordinate systems}
\label{sec:smooth-vs-discrete}

By Theorem~\ref{thm:binets-docs}, there is a close correspondence
between three-dimensional principal binets $b : D_3 \rightarrow \R^3$
and discrete orthogonal coordinate systems.
In an appropriate continuum limit, the latter should converge to smooth orthogonal coordinate systems.
Hence, three-dimensional principal binets can be viewed as a discretization of smooth three-dimensional orthogonal coordinate systems.

On the other hand, as discussed in Remark~\ref{rem:focal-points}, given a conjugate binet $b : D_3 \rightarrow \R^3$,
one may regard the points on the vertices $V_3$ together with the points on the 12-faces $F_3^{12}$ as discretizations of the coordinate system
(in the sense that these points should converge to a smooth coordinate system in the continuum limit).
From this perspective, the points on the $13$-face $F_3^{13}$ and the $23$-faces $F_3^{23}$ play the role of focal points.
Nevertheless, these focal points are still involved in the orthogonality constraints of the principal binet (Definition~\ref{def:principal-binet}).
In this section we investigate how the orthogonality conditions on the focal points still lead to an orthogonal coordinate system in the smooth case.

In the smooth case, a triply conjugate system $x : U \subset \R^3 \rightarrow \R^3$ is governed by three equations of the form \cite{darbouxlecons, ddgbook}
\begin{equation}
  \label{eq:tcs}
  \partial_i \partial_j x = a_{ij}\,\partial_i x + a_{ji}\,\partial_j x
\end{equation}
for distinct $i, j \in \{1, 2, 3\}$,
with six smooth functions $a_{ij} : U \rightarrow \R$ satisfying three compatibility conditions
\begin{equation}
  \label{eq:tcs-compatibility}
  \partial_k a_{ij} = \partial_j a_{ik} = a_{ij} a_{jk} + a_{kj} a_{ik} - a_{ij} a_{ik}.
\end{equation}

The \emph{Laplace transforms} (or \emph{focal points}) of $x$ in $i$-direction along the $ij$-surfaces are given by
\[
  L_{ij} = x - \frac{1}{a_{ji}} \, \partial_i x.
\]
We assume that the coefficients $a_{ij}$ do not vanish, so that the Laplace transforms are well defined,
and that the transforms are non-degenerate, in particular, that their derivatives do not vanish.

\begin{figure}[t]
  \centering
  \begin{overpic}[width=0.4\textwidth]{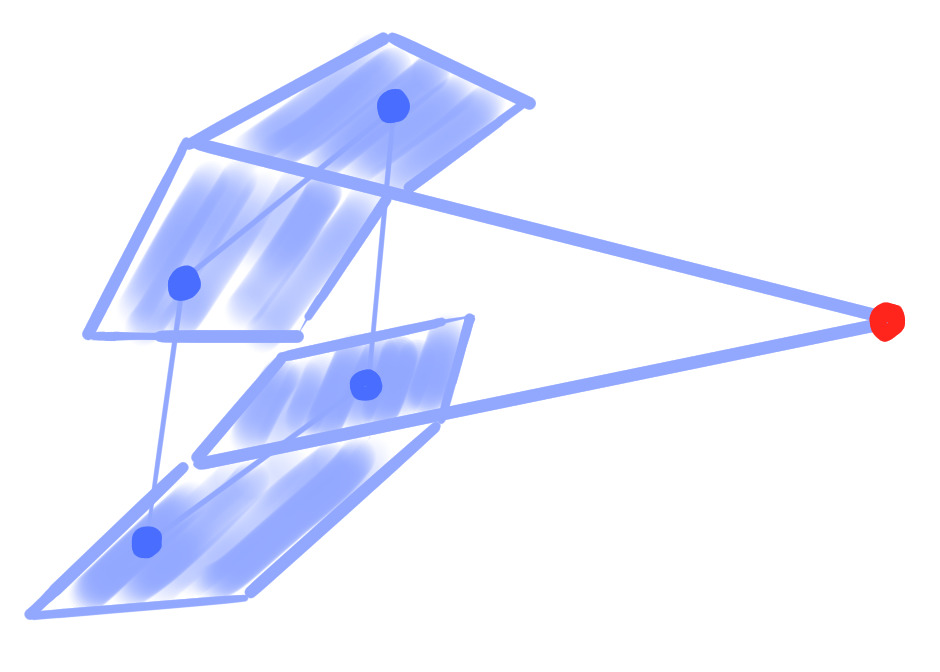}
    \put(95,30){$\color{red}b(f)$}
    \put(19,9){\contour{white}{$\color{blue}\square b(v)$}}
    \put(42,33){\contour{white}{$\color{blue}\square b(v_2)$}}
    \put(-1,45){\contour{white}{$\color{blue}\square b(v_3)$}}
    \put(18,64){\contour{white}{$\color{blue}\square b(v_{23})$}}
  \end{overpic}
  \hspace*{0.05\textwidth}
  \begin{overpic}[width=0.4\textwidth]{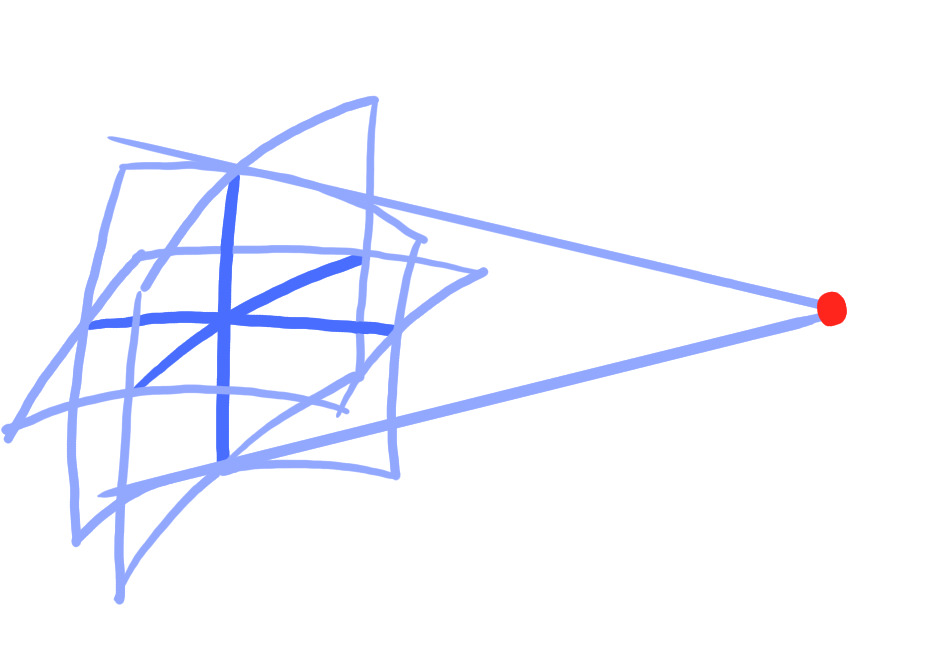}
    \put(89,32){$\color{red}L_{13}$}
    \put(41,34.5){\contour{white}{$\color{blue}1$}}
    \put(38,43){\contour{white}{$\color{blue}2$}}
    \put(23.5,53){\contour{white}{$\color{blue}3$}}
  \end{overpic}
  \caption{
    \emph{Left}: Point on a $23$-face of a conjugate binet.
    \emph{Right}: Focal point $L_{13}$ of a triply conjugate system.
  }
  \label{fig:focal-points}
\end{figure}
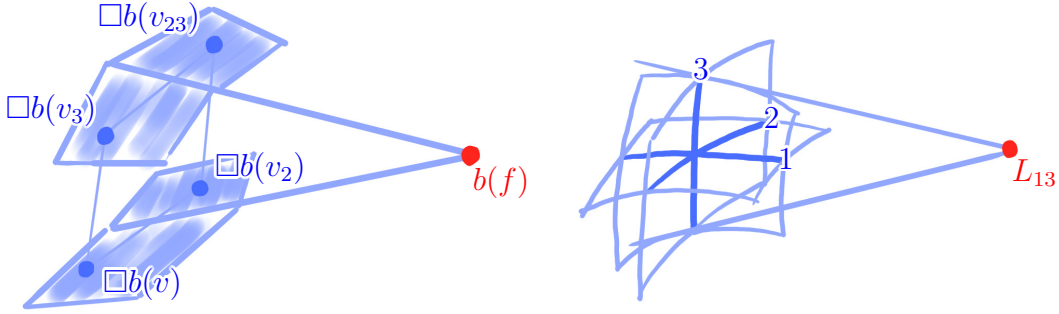

Now, we identify which focal points are discretized by the points on the $13$-faces and $23$-faces.
Consider a $23$-face $f = (v, v_2, v_{23}, v_3) \in F_3^{23}$.
The two intersection lines $\square b(v) \cap \square b(v_2)$ and $\square b(v_3) \cap \square b(v_{23})$ discretize tangent lines in the $1$-direction.
Consequently, the point of intersection $b(f)$ discretizes a focal point $L_{13}$ in the $1$-direction along the $13$-surfaces (see Figure~\ref{fig:focal-points}).
Similarly, a point $b(f)$ on a $13$-face $f \in F_3^{13}$ discretizes a focal point $L_{23}$ in the $2$-direction along the $23$-surfaces.

In terms of orthogonality conditions for a principal binet,
we distinguish the following two types, depending on the directions of edges in $E_3$:
\begin{itemize}
\item Crosses involving edges in $1$-direction or $2$-direction:
  It suffices to consider crosses of the form
  $(v, f, v', f') \in C_3$ with $f, f' \in F_3^{12}$.
  Crosses with edges in $1$-direction or $2$-direction involving $13$-faces or $23$-faces yield dependent conditions due to Remark~\ref{rem:principal-constraints}.
  In the smooth limit, these correspond to the orthogonality condition $\langle \partial_1 x,\, \partial_2 x \rangle = 0$.
\item Crosses involving edges in $3$-direction:
  It suffices to consider crosses of the form
  $(v, f, v', f') \in C_3$ with $v' = v + e_3$ and $f, f' \in F_3^{23}$.
  Again, crosses with edges in $3$-direction involving $13$-faces yield dependent conditions due to Remark~\ref{rem:principal-constraints}.
  In the smooth limit, these correspond to the orthogonality condition $\langle \partial_2 L_{13},\, \partial_3 x \rangle = 0$.
\end{itemize}

It turns out, that also in the smooth case these two orthogonality constraints characterize orthogonal coordinate systems.
\begin{theorem}
  \label{thm:tos-from-focal-points}
  A triply conjugate system $x$ is orthogonal
  if and only if
  \begin{equation}
    \sca{\partial_1 x, \partial_2 x} = 0
    \quad\text{and}\quad
    \sca{\partial_2 L_{13}, \partial_3 x} = 0.
  \end{equation}
\end{theorem}

\begin{proof}
  Assume $\sca{\partial_1 x, \partial_2 x} = 0$.

  First, observe that the remaining two orthogonality conditions for an orthogonal system,
  $\sca{\partial_2 x, \partial_3 x} = 0$ and $\sca{\partial_3 x, \partial_1 x} = 0$,
  are not independent.
  Indeed, differentiating the first condition gives
  \begin{equation}
    \label{eq:orthogonality-dependence}
    0
    = \partial_3 \sca{\partial_1 x, \partial_2 x}
    = \sca{\partial_1 \partial_3 x, \partial_2 x} + \sca{\partial_1 x, \partial_2 \partial_3 x}
    = a_{31} \sca{\partial_2 x, \partial_3 x} + a_{32} \sca{\partial_1 x, \partial_3 x},
  \end{equation}
  where we used \eqref{eq:tcs}.

  Next, consider the derivative of the Laplace transform:
  \begin{equation}
    \begin{aligned}
      \partial_2 L_{13}
      &= \partial_2\left( x - \frac{1}{a_{31}} \partial_1 x \right) \\
      &= \partial_2 x + \frac{\partial_2 a_{31}}{a_{31}^2} \partial_1 x - \frac{1}{a_{31}} \partial_1 \partial_2 x \\
      &= \frac{\partial_2 a_{31} - a_{12} a_{31}}{a_{31}^2} \partial_1 x + \frac{a_{31} - a_{21}}{a_{31}} \partial_2 x \\
    \end{aligned}
  \end{equation}
  where we used \eqref{eq:tcs}.
  Now using \eqref{eq:tcs-compatibility}, we obtain
  \begin{equation}
    \partial_2 L_{13} = \frac{a_{21} - a_{31}}{a_{31}} \left( \frac{a_{32}}{a_{31}} \partial_1 x - \partial_2 x \right).
  \end{equation}
  By assumption, the Laplace transform is non-degenerate, in particular $a_{21} \neq a_{31}$.
  Hence, the scalar product with $\partial_3 x$ reads
  \begin{equation}
    \sca{\partial_2 L_{13}, \partial_3 x}
    = \frac{a_{21} - a_{31}}{a_{31}} \left( \frac{a_{32}}{a_{31}} \sca{\partial_1 x, \partial_3 x} - \sca{\partial_2 x, \partial_3 x} \right).
  \end{equation}
  If $x$ is orthogonal, this is indeed zero.
  Conversely, if this scalar product vanishes, then by \eqref{eq:orthogonality-dependence} and our assumptions
  $a_{31}, a_{32} \neq 0$ and $a_{21} \neq a_{31}$, we deduce $\sca{\partial_2 x, \partial_3 x} = \sca{\partial_3 x, \partial_1 x} = 0$.
\end{proof}

\bibliographystyle{alpha}

\bibliography{references}

@misc{atbinets,
    title="{Principal binets}",
      author={Niklas C. Affolter and Jan Techter},
      year={2025},
    note={To appear in Discrete \& Computational Geometry}
}

@article{bpisothermic,
	author = {Pinkall, Ulrich and Bobenko, Alexander},
	journal = {Journal für die reine und angewandte Mathematik},
	pages = {187-208},
	title = {Discrete isothermic surfaces.},
	url = {http://eudml.org/doc/153826},
	volume = {475},
	year = {1996},
}

@inbook{bobenkocircular,
	place={Cambridge}, 
	series={London Mathematical Society Lecture Note Series}, 
	title={Discrete conformal maps and surfaces}, 
	DOI={10.1017/CBO9780511569432.009},
	booktitle={Symmetries and Integrability of Difference Equations}, 
	publisher={Cambridge University Press}, 
	author={Bobenko, Alexander I.}, 
	year={1999}, 
	pages={97–108}, 
	collection={London Mathematical Society Lecture Note Series}
}

@article{bslinecomplexes,
	author = {Bobenko, Alexander I. and Schief, Wolfgang K.},
	title = "{Discrete line complexes and integrable evolution of minors}",
	journal = {Proceedings of the Royal Society A},
	volume = {471},
	number = {2175},
	year = {2015},
	month = {03},
	doi = {10.1098/rspa.2014.0819},
	url = {https://doi.org/10.1098/rspa.2014.0819},
}

@article{bsorganizing,
	doi = {10.1070/rm2007v062n01abeh004380},
	url = {https://doi.org/10.1070%2Frm2007v062n01abeh004380},
	year = 2007,
	month = {02},
	publisher = {{IOP} Publishing},
	volume = {62},
	number = {1},
	pages = {1--43},
	author = {Alexander I. Bobenko and Yuri B. Suris},
	title = "{On organizing principles of discrete differential geometry. Geometry of spheres}",
	journal = {Russian Mathematical Surveys},
}

@book{ddgbook,
	title={Discrete Differential Geometry: Integrable Structure},
	author={Bobenko, Alexander I. and Suris, Yuri B.},
	isbn={978-0-8218-4700-8},
	series="{Graduate Studies in Mathematics}",
	volume="98",
	year={2008},
	publisher={American Mathematical Society}
}

@article{bsstconfocali,
	author = {Bobenko, Alexander I. and Schief, Wolfgang K. and Suris, Yuri B. and Techter, Jan},
	title = "{On a discretization of confocal quadrics. I. An integrable systems approach}",
	journal = {Journal of Integrable Systems},
	volume = {1},
	number = {1},
	pages = {xyw005},
	year = {2016},
	month = {08},
	issn = {2058-5985},
	doi = {10.1093/integr/xyw005},
	url = {https://doi.org/10.1093/integr/xyw005},
	eprint = {https://academic.oup.com/integrablesystems/article-pdf/1/1/xyw005/19520676/xyw005.pdf},
}

@article{bsstconfocalii,
	author = {Bobenko, Alexander I. and Schief, Wolfgang K. and Suris, Yuri B. and Techter, Jan},
	title = "{On a Discretization of Confocal Quadrics. A Geometric Approach to General Parametrizations}",
	journal = {International Mathematics Research Notices},
	volume = {2020},
	number = {24},
	pages = {10180-10230},
	year = {2018},
	month = {12},
	issn = {1073-7928},
	doi = {10.1093/imrn/rny279},
	url = {https://doi.org/10.1093/imrn/rny279},
	eprint = {https://academic.oup.com/imrn/article-pdf/2020/24/10180/34948952/rny279.pdf},
}

@article{cdscircular,
	title = "The integrable discrete analogues of orthogonal coordinate systems are multi-dimensional circular lattices",
	journal = "Physics Letters A",
	volume = "235",
	number = "5",
	pages = "480 - 488",
	year = "1997",
	issn = "0375-9601",
	doi = "https://doi.org/10.1016/S0375-9601(97)00657-9",
	url = "http://www.sciencedirect.com/science/article/pii/S0375960197006579",
	author = "Jan Cieśliński and Adam Doliwa and Paolo M. Santini"
}

@book{darbouxlecons,
  author    = {Gaston Darboux},
  title     = {Leçons sur la théorie générale des surfaces et les applications géométriques du calcul infinitésimal},
  edition   = {2},
  publisher = {Gauthier‐Villars},
  address   = {Paris, France},
  year      = {1915},
  note      = {Deuxième édition revue et augmentée},
}

@Article{dellingercbpatterns,
	author={Dellinger, Felix},
	title="{Discrete Isothermic Nets Based on Checkerboard Patterns}",
	journal={Discrete {\&} Computational Geometry},
	year={2023},
	month={Sep},
	day={14},
	issn={1432-0444},
	doi={10.1007/s00454-023-00558-1},
	url={https://doi.org/10.1007/s00454-023-00558-1}
}

@article{doliwaqnetsquadrics,
	title = {Quadratic reductions of quadrilateral lattices},
	journal = {Journal of Geometry and Physics},
	volume = {30},
	number = {2},
	pages = {169-186},
	year = {1999},
	issn = {0393-0440},
	doi = {https://doi.org/10.1016/S0393-0440(98)00053-9},
	url = {https://www.sciencedirect.com/science/article/pii/S0393044098000539},
	author = {Adam Doliwa},
}

@article{doliwaanetspluecker,
	title = "{Discrete asymptotic nets and W-congruences in Plücker line geometry}",
	journal = "Journal of Geometry and Physics",
	volume = "39",
	number = "1",
	pages = "9 - 29",
	year = "2001",
	issn = "0393-0440",
	doi = "https://doi.org/10.1016/S0393-0440(00)00070-X",
	url = "http://www.sciencedirect.com/science/article/pii/S039304400000070X",
	author = "Adam Doliwa"
}

@article{dsqnet,
	title = "Multidimensional quadrilateral lattices are integrable",
	journal = "Physics Letters A",
	volume = "233",
	number = "4",
	pages = "365 - 372",
	year = "1997",
	issn = "0375-9601",
	doi = "https://doi.org/10.1016/S0375-9601(97)00456-8",
	author = "Adam Doliwa and Paolo M. Santini"
}

@article{dsmlinecongruence,
	author = {Doliwa, Adam and Santini, Paolo M. and Mañas, Manuel},
	title = {Transformations of quadrilateral lattices},
	journal = {Journal of Mathematical Physics},
	volume = {41},
	number = {2},
	pages = {944-990},
	year = {2000},
	doi = {10.1063/1.533175},	
	URL = {https://doi.org/10.1063/1.533175},	
}

@article{lpwywconical,
	title = "Geometric Modeling with Conical Meshes and Developable Surfaces",
	author = "Yang Liu and Helmut Pottmann and Johannes Wallner and Yongliang Yang and Wenping Wang",
	year = "2006",
	month = "7",
	day = "1",
	doi = "10.1145/1141911.1141941",
	volume = "25",
	pages = "681--689",
	journal = "ACM Transactions on Graphics",
	issn = "0730-0301",
	publisher = "Association for Computing Machinery",
	number = "3",
}

@Article{sauerqnet,
	author={Sauer, Robert},
	title="{Wackelige Kurvennetze bei einer infinitesimalen Fl{\"a}chenverbiegung}",
	journal={Mathematische Annalen},
	year=1933,
	volume=108,
	number=1,
	pages={673-693},
	issn={1432-1807},
	doi={10.1007/BF01452858},
	url={https://doi.org/10.1007/BF01452858}
}

@phdthesis{Techter2021,
  author    = {Techter, Jan},
  title     = {Discrete confocal quadrics and checkerboard incircular nets},
  school    = {Technische Universit\"at Berlin},
  year      = 2021,
  doi       = {10.14279/depositonce-11461},
  url       = {https://depositonce.tu-berlin.de/items/62d4782e-82a4-40fb-9cfa-1526542b9b4b}
}

\end{document}